\documentclass[11pt]{article}
\usepackage{geometry}
\usepackage{fancyhdr, titling}

\usepackage{times}
\usepackage{silence}
\ActivateWarningFilters[pdftoc]

\usepackage{common}
\title{Active multiple testing with proxy p-values and e-values}
\author{
\hspace{1.5in}
Ziyu Xu$^1$
\and
Catherine Wang$^1$ \hspace{1.5in}
\and
Larry Wasserman$^{1, 2}$
\and
Kathryn Roeder$^{1, 3}$
\and
Aaditya Ramdas$^{1, 2}$
\and \vspace{0.05in} \\
\texttt{\{xzy,catheri2,larry,roeder,aramdas\}@cmu.edu}
\vspace{0.1in} \\ ${}^{1}$Department of Statistics and Data Science\\
${}^{2}$Machine Learning Department\\
${}^{3}$Computational Biology Department\\
Carnegie Mellon University
}
\date{\today}

\begin{document}
\maketitle

\begin{abstract}
Researchers often lack the resources to test every hypothesis of interest directly or compute test statistics comprehensively, but often possess auxiliary data from which we can compute an estimate of the experimental outcome.
We introduce a novel approach for selecting which hypotheses to query a statistic (e.g., run an experiment, perform expensive computation, etc.) in a hypothesis testing setup by leveraging estimates to compute proxy statistics.
Our framework allows a scientist to propose a proxy statistic and then query the true statistic with some probability based on the value of the proxy. We make no assumptions about how the proxy is derived, and it can be arbitrarily dependent on the true statistic. If the true statistic is not queried, the proxy is used in its place. 
We characterize "active" methods that produce valid p-values and e-values in this setting and utilize this framework in the multiple testing setting to create procedures with false discovery rate (FDR) control.
Through simulations and real data analysis of causal effects in scCRISPR screen experiments, we empirically demonstrate that our proxy framework has both high power and low resource usage when our proxies are accurate estimates of the respective true statistics. \end{abstract}

\tableofcontents

\section{Introduction}
Modern science often engages in large-scale hypothesis testing (e.g., many drugs to test in animal trials, many proteins to check if they fold correctly and attach to a binding site, etc.) that is resource-constrained to only a limited number of experiments.
As a result, researchers often do not have the material resources to test every single hypothesis and have to pick a subset of hypotheses to test. However, they would still like to make as many true discoveries as possible.
Often, we have some prior information (e.g., human experts, a pretrained machine learning model, etc.) that can predict the outcome of an experiment. Thus, our goal in this paper is to develop a framework for utilizing such prior information to guide which hypotheses to test, while still being able to make valid inferences about all hypotheses, including the ones that were not directly tested.

The setting we will consider in this paper is as follows. Suppose we wish to test a hypothesis about the distribution of $Z$. However, $Z$ is ``expensive'' to obtain (e.g., requires a lot of computation, consumes extra experimental resources, can only be obtained at a later time, etc.). We assume that we also have access to another piece of data $X$, that may be dependent on $Z$. We assume that $X$ is ``cheap'' to obtain, and thus we always have access to it.
Using $X$, we produce a ``proxy statistic'' $\hat{S} \coloneqq \hat{S}(X)$ which is an approximation of what the true statistic $S \coloneqq S(X, Z)$ would be, if an experiment were to be run to collect $Z$. We assume that $\hat{S}$ lies in the same domain as its corresponding true statistic (i.e., it lies in $[0, 1]$ for p-values and is nonnegative for e-values). Since it is only an approximation, we do not make any assumptions about its validity, e.g., a proxy p-value may be larger than uniform under the null, and a proxy e-value may have an expectation greater than 1. Further, we allow for any dependence structure between the variables $X$ and $Z$ (and consequently $\hat{S}$ and $S$). The question is: how can we use the predictions $\hat{S}$ to choose whether or not to collect $S$, and regardless of that choice, always end up with a valid statistic (i.e., a valid p-value or e-value)?

\begin{figure}[ht]
\includegraphics[width=\textwidth]{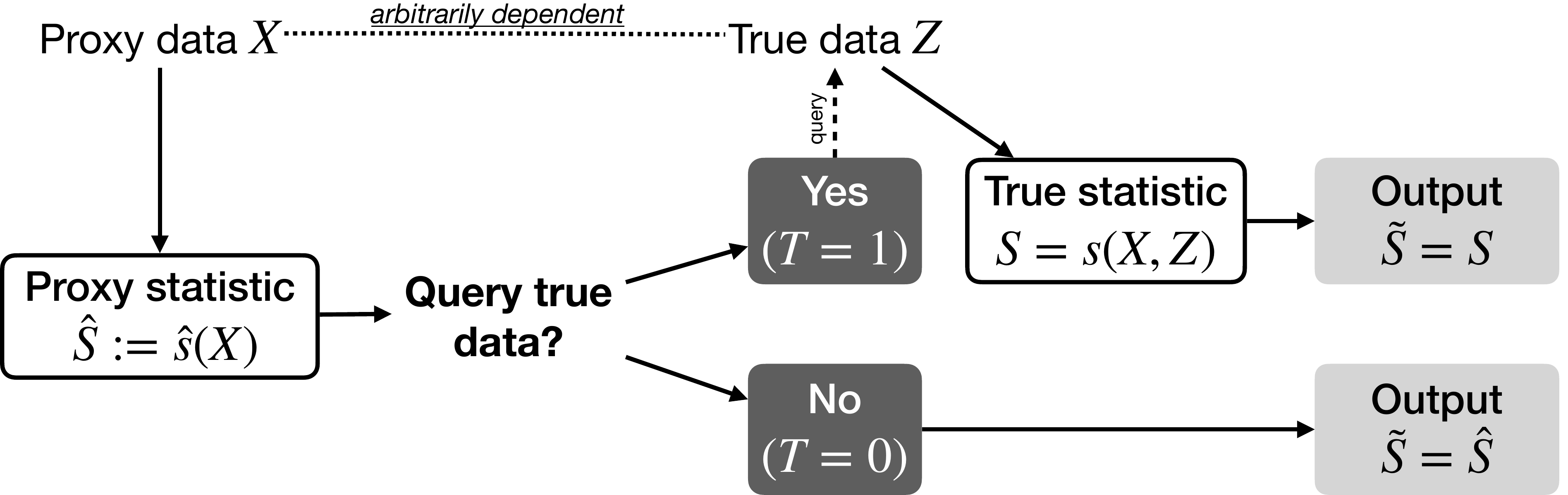}
\caption{Diagram of the active hypothesis testing setup with proxy statistics. The proxy data $X$ is used to calculate a proxy statistic and decide whether the true data, $Z$, should be queried.}
\end{figure}

\paragraph{Outline} We present a simple solution to the above problem for deriving both what we call \emph{active} e-values and p-values, which query the true data in a probabilistic fashion depending on the value of the proxy statistic. We use ``active'' in the sense that the user actively chooses whether to query the true statistic or not based on the proxy.
We introduce the active e-values and p-values and motivating applications for their usage in the remainder of this section.
In \Cref{sec:active-stats}, we prove the validity of these active statistics. We then formulate applications of active e-values and p-values for multiple testing with false discovery rate control (FDR) in \Cref{sec:fdr-control}.
We include some numerical simulations demonstrating the empirical performance of our active methods in \Cref{sec:simulations}.
We then show the utility of our active methods when computing p-values for detecting causal effects in scCRISPR screen experiments in \Cref{sec:scCRISPR}.
Lastly, we discuss an extension of our active framework when we can estimate or access the joint distribution of the proxy and true statistics under the null in \Cref{sec:joint}.
\subsection{Active e-values and p-values: debiasing proxy statistics using randomization}
We will now formally introduce the active e-value and p-value constructions that we will use throughout this paper.

E-values are a category of statistics that have been the subject of much recent interest and have enabled innovations in composite hypothesis testing, anytime-valid inference, multiple testing, and other areas --- \citet{ramdas_gametheoretic_statistics_2022} provides an overview of these developments.
Let $F$ denote a proxy e-value and $E$ be the true e-value; hence, we make the following assumptions.
\begin{align}
    F\text{ and }E\text{ are nonnegative}.\ E \text{ is a bona-fide e-value, i.e., }\expect[E] \leq 1\text{ under the null }H_0.
\end{align}
Clearly, there is no distributional assumption on $F$: it can be far from an e-value.
We do not restrict the dependence between $F$ and $E$, i.e., they can be \emph{arbitrarily dependent}.
As a result, we cannot use the proxy $F$ directly as a test statistic since it is not valid, unlike $E$.
Yet, we can formulate an active e-value that is valid and uses $F$ to inform how often to query $E$.
We simply choose to query $E$ with probability $(1 - \gamma F^{-1})_+$. That is, if $F$ is large, we are more likely to query $E$ --- here $\gamma \in (0, 1]$ is a user-chosen parameter. Formally, let $T$ denote the indicator variable of querying the true e-value $E$. Then,
\[
T \mid F  \sim \text{Bern}(1 - \gamma F^{-1})_+,
\]
where $x_+$ denotes the maximum of $x$ and 0.
We then define an \emph{active e-value}
\begin{align}
    \tilde{E} \coloneqq (1 - T) \cdot F + T \cdot \left(1 - \gamma\right) \cdot E.
    \label{eq:active-e-value}
\end{align}
Thus, we get that $\tilde{E}$ is the proxy $F$ with probability $(\gamma F^{-1}) \wedge 1$ and the true e-value $E$ with probability $(1 - \gamma F^{-1})_+$.

We can also define an active p-value under this scheme. Assume we have access to a proxy p-value, $Q$, that has no distributional assumptions placed on it, and is arbitrarily dependent on a p-value $P$. Formally, we assume the following:
\begin{align}
Q\text{ and }P\text{ are supported on }[0, 1]\text{ and }
\prob{P \leq s} \leq s\text{ for all }s \in [0, 1],
\text{ under the null }H_0.
\end{align}
Then, we can define the following active p-value:
\begin{align}
    \tilde{P}\coloneqq  (1 - T)Q + T (1 - \gamma)^{-1} \cdot P  \label{eq:arbdep-active-pvalue},
\end{align} where $T \mid Q \sim \text{Bern}(1 - \gamma Q)$ and $\gamma \in [0, 1)$.
Similarly to the e-value case, $\tilde{P}$ equals the proxy $Q$ with probability $(\gamma Q) \wedge 1$ and the true p-value with probability $(1 - \gamma Q)_+$. Another connection to note is that since the reciprocal of a e-value is a p-value (by Markov's inequality), an active p-value with $Q = F^{-1}$ and $P = E^{-1}$ is identical to the reciprocal of the active e-value formed by $F$ and $E$.

Finally, we present another active p-value, where we assume independence between $Q$ and $P$ and that we know the density function of the p-value $Q$ under the null hypothesis, e.g., we can estimate it well from existing data. Let this density function be denoted as $f$ and let $\ell_f \coloneqq \inf_{q \in [0, 1]} f(q)$.
We define an active p-value as follows:
\begin{align}
    \tilde{P}^{\density} \coloneqq T P + (1 - T)Q \text{ where } T|Q \sim \textnormal{Bern}\left(1 - \frac{\eta \ell_f}{f(Q)}\right).\label{eq:density-active-pvalue}
\end{align}
Thus, $\tilde{P}^{\density}$ equals the true p-value $P$ when we query $(T=1)$ and the proxy p-value $Q$ when we do not query $(T=0)$. 

\paragraph{Generality of active statistics} Since $\tilde{E}$ and $\tilde{P}$ (and $\tilde{P}^\density$ under additional assumptions) possess validity properties of being either an e-value or a p-value, we can utilize them in any existing procedure that can provide valid error control guarantees when provided valid input e-values and p-values. In this paper, we will focus on their utilization in variants of the Benjamini-Hochberg (BH) procedure for FDR control.

However, other multiple testing procedures exist which control the family-wise error rate (\citealp{marcus_closed_testing_1976}; \citealp{holm_simple_sequentially_1979}; \citealp{hochberg_sharper_bonferroni_1988}), or provide simultaneous false discovery proportion bounds (\citealp{genovese_stochastic_process_2004,genovese_exceedance_control_2006}; \citealp{goeman_multiple_testing_2011}, \citealp{goeman_only_closed_2021}). These methods all rely on developing a statistic for testing the global null. In the e-value case, it is relatively simple to construct an e-value under the global null --- if $(E_1, \dots, E_K)$ are independent or sequentially dependent, then the product $\tilde{E} = \prod_{i = 1}^K \tilde{E}_i$ is a valid e-value under the global null. If they are arbitrarily dependent, the average $\tilde{E} = K^{-1} \sum_{i = 1}^K \tilde{E}_i$ is a valid e-value. One can also accomplish the same with p-values by using p-merging functions to derive a p-value for testing the global null \citep{simes_improved_bonferroni_1986,vovk_combining_p-values_2020,vovk_admissible_ways_2022}.
Thus, these methods that generally rely on e-values and p-values can also utilize the active statistics we have formulated.

\subsection{Motivating applications}

The following are some scenarios where a proxy ($F$ or $Q$) may be cheaply obtained, and occasionally querying the true statistic ($E$ or $P$) saves resources.

\begin{enumerate}[leftmargin=*]

    \item \emph{Screening for causal effects using non-causal estimates.} In gene perturbation experiments, we aim to test the effect of perturbing one gene on the expression of other genes. Unfortunately, there is often unmeasured confounding in these experiments that can arise from multiple sources, e.g., samples from different lab environments, with different cell characteristics, etc.
        We can measure the expression levels of thousands of genes, and we wish to test the causal effect of different gene perturbations on these genes. As a result, we often need to test at least tens of thousands of hypotheses. A recent line of work in proximal causal inference \citep{cui_semiparametric_2023,miao_confounding_2020,liu_regression-based_2024} has developed methods of using negative controls to account for such confounding.
        However, computing the p-values for causal effects using proximal causal inference is quite computationally expensive. To alleviate this burden, one can use the biased ordinary least squares p-value (i.e., a p-value for the coefficient in a regression model for gene expression) as the proxy $Q$. We demonstrate in \Cref{sec:scCRISPR}, in an application to real genetic data, that our active framework preserves validity in this setting and greatly reduces the computation time while maintaining power.

    \item \emph{Selecting hypotheses for experimental study.} In many domains, we can obtain observational and potentially corrupted data concerning some underlying process, e.g., consumer behavior in economics, brain fMRI data in neuroscience, surveys on psychological traits, etc.
       Hence, one can let the proxy be an estimate of the true statistic that is learned from existing observational data, and then actually conduct randomized experiments to calculate the proxy for only the promising hypotheses, while still deriving a valid statistic for every hypothesis of interest, including the ones that were never directly tested.

    \item  \emph{Clinical endpoint selection from patient attributes.} In running preliminary drug trials, a researcher often wishes not only to test the efficacy of a drug in its primary use, but also potential secondary endpoints as well.
        In this situation, we can use the demographics of the current patient sample to produce a proxy that estimates the statistical evidence for each endpoint --- then we can randomly select to only compute true statistics for the hypotheses that have large $F$ (and perform the necessary medical procedure to extract the necessary data, e.g., tissue sample, X-ray, etc.).

    \item \emph{Variable selection with expert advice.} Variable selection has been a core problem of statistics, and typically, one uses existing data to determine which covariates are important for predicting an outcome of interest. However, one often has domain knowledge about what these covariates are, and we can elicit predictions from either human domain experts or large language models \citep{jeong_llm-select_feature_2024} to guide the selection of important covariates. This can be used to inform the types of features we may wish to collect to improve the models.

    \item \emph{Early proxies for future outcomes.} A common paradigm in experiments is that units are exposed to a treatment, and their response is monitored over time.
    This occurs in many experimentation scenarios, such as A/B testing, lifetime outcomes of policy interventions, and analysis of therapeutic efficacy in clinical trials.
    Since one does not always wish to wait the entire duration before making a judgment on the hypothesis of interest, there has been significant interest in understanding how to use proxies, i.e., observations of earlier features and response of each unit to improve estimation of the long-term effects \citep{athey_surrogate_index_2019,athey_combining_experimental_2020,tran_inferring_long-term_2024,prentice_surrogate_endpoints_1989}.
    In this context, we can view the proxy as an estimate of the true statistic, since it simply uses the outcome process at an earlier time. Moreover, our active framework can handle the fact that the earlier outcome is dependent on the final outcome.
    By randomly selecting which experiments to wait out (potentially for years) to see the long-term results, we can conclude early on hypotheses not selected for long-term observation using our proxies.

    \item \emph{Computationally expensive e-values.} Consider the setting when $Z$ is known ahead of time, but utilizing it to actually compute $E$ may be quite expensive.
    When a scientist's computational budget is limited, it would be prudent to use little computation to derive a rough estimate, $F$, for each hypothesis, and only utilize increased computation to calculate a powerful $E$ for those that had a large $F$ already. For example, the universal inference e-value \citep{wasserman_universal_inference_2020} requires one to compute the maximum likelihood over the set of null distributions. When testing null hypotheses such as the family of log-concave distributions \citep{dunn_universal_inference_2022} or Gaussian mixture models \citep{shi_universal_inference_2024}, finding this maximum likelihood requires an expensive nonconvex optimization procedure.
    Further, one usually does not compute the universal inference e-value once, but takes the average over multiple data splits \citep{dunn_universal_inference_2022,dunn_gaussian_universal_2022}. Thus, recomputing the maximum likelihood under the null in these settings can become quite expensive.
\end{enumerate}

\paragraph{Related work.} Our method is similar in spirit to two-stage multiple testing methods.
Two-stage multiple testing is often used in settings where the number of null hypotheses is expected to vastly outstrip the number of non-nulls to increase the power of a multiple testing procedure.
In the first stage, the scientist first filters out a large number of candidate hypotheses before performing the actual multiple testing procedure on the remaining hypotheses in the second stage.
The first stage can involve a test that utilizes a fraction of samples (\citealp{zehetmayer_two-stage_designs_2005}; \citealp{aoshima_two-stage_procedures_2011}) before allocating the remaining samples to the hypotheses still in consideration. Another approach, often seen in testing gene-gene interactions in genomics applications, utilizes an alternative statistic (e.g., a p-value for the marginal effect of a single gene) that is independent of the true statistic used in the second stage (e.g., a p-value for the effect of a gene-gene interaction) to reduce the number of hypotheses in consideration (\citealp{kooperberg_increasing_power_2008}; \citealp{murcray_gene-environment_interaction_2009}; \citealp{gauderman_efficient_genome-wide_2010}; \citealp{dai_two-stage_testing_2012}; \citealp{lewinger_efficient_two-step_2013}; \citealp{pecanka_powerful_efficient_2017}; \citealp{kawaguchi_improved_two-step_2023}).
In either case, these two-stage methods require that the statistic used to filter hypotheses in the first stage is independent of the second stage, or the conditional distribution of the second stage statistic on the first stage statistic is known.
This is accomplished through adaptive data allocation, parametric assumptions, or by relying on the asymptotic distribution being a multivariate Gaussian.
Notably, many two-stage methods are specific to gene-gene or gene-environment settings since the first stage statistic specifically leverages the marginal effect of a single gene to filter out pairs.
Both of these requirements significantly reduce the applicability of two-stage methods, particularly in modern settings where only nonparametric assumptions are made about the collected data.
In contrast, our active multiple testing procedures we introduce in this paper, are virtually assumption-free w.r.t.\ the proxies, i.e., the statistic used for filtering the hypotheses in the first stage.
Further, our approach is generic --- we can apply our methodology to any multiple testing problem where there exists an imperfect proxy and the ability to derive bona fide p-values or e-values.

Another branch of work that is close to ours in spirit is the field of active learning, used for actively selecting which data to label for machine learning
(\citealp{balcan_agnostic_active_2006};  \citealp{joshi_multi-class_active_2009}; \citealp{hanneke_theory_disagreement-based_2014}; \citealp{gal_deep_bayesian_2017}; \citealp{ash_deep_batch_2019};  \citealp{ren_survey_deep_2021}; \citealp{vishwakarma_promises_pitfalls_2023}; \citealp{cheng_how_many_2024}).
Recently, \citet{zrnic_active_statistical_2024a} presented a framework for active statistical inference in the context of data labeling, although it is for estimating a single parameter or testing a single hypothesis, as opposed to the multiple hypothesis testing setting we consider.

\section{The validity of active statistics\label{sec:active-stats}}
We will now prove that the active e-value and both active p-values are valid: a valid e-value and a pair of valid p-values, respectively. The $\tilde{E}$ and $\tilde{P}$ active statistics are robust to the dependence structure between the proxy and the true statistics. On the other hand, we will show $\tilde{P}^\density$ is valid under independence between $P$ and $Q$.
\subsection{Active statistics valid under arbitrary dependence\label{sec:active-dep}}
We first prove that the validity of the active e-value is robust to dependence assumptions.
\begin{proposition}
    If $E$ is a valid e-value and $F$ is an arbitrary nonnegative random variable, then under any dependence structure between $F$ and $E$, the active e-value $\tilde{E}$ in \eqref{eq:active-e-value} is a valid e-value.
    \label{prop:active-e-value}
\end{proposition}
\begin{proof}
    When the null hypothesis is true, we can compute the expectation of $\tilde E$ in the following way:
    \begin{align}
        \expect[\tilde E] = \expect[(1 - T) \cdot F] + \expect\left[T \cdot \left(1 - \gamma\right) \cdot E \right] \leq \gamma + \left(1 - \gamma\right) \cdot \expect[E]  \leq 1.
    \end{align} The first inequality is because $\expect[T \mid E, F] = (1 - \gamma F^{-1})_+ \leq 1$. The second inequality is because $E$ is an e-value. Thus, we have shown that $\tilde{E}$ is a bona fide e-value.
\end{proof}

We can also show this kind of validity for the active p-value, and defer the proof of the following proposition to \Cref{sec: active arb dep p-value valid proof}.
\begin{proposition}\label{prop: active arb dep pvalue valid}
    If $P$ is a valid p-value and $Q$ is an arbitrary $[0,1]$-valued random variable, then under any dependence structure between $P$ and $Q$, the active p-value $\tilde{P}$ in \eqref{eq:arbdep-active-pvalue} is a valid p-value.
\end{proposition}
\subsection{An active p-value under a known proxy distribution\label{sec:active-ind-pvalue}}
We consider an alternative formulation of the active p-value under the stronger assumptions of (1) a known or estimable density of $Q$, and (2) independence between $Q$ and $P$.
While one can use independent proxies to filter hypotheses in a two-stage multiple testing setting (i.e., use the proxy to select hypotheses, and then only perform a standard multiple testing procedure on the true statistics for the selected hypotheses), our goal with the active statistic framework is to provide an informative statistic even for hypotheses where the true statistic is not queried. Further, this also allows us to decouple the active component from any specific multiple testing procedure, and use it with any procedure that accepts p-values (or e-values) as input.

\begin{proposition}\label{prop:density-act-pvalue-valid}
    When $P$ and $Q$ are independent, $\tilde{P}^{\density}$ in \eqref{eq:density-active-pvalue} is a valid p-value. Further, $\tilde{P}^{\density}$ is exactly uniform if $P$ is exactly uniform, regardless of the distribution of $Q$.
    Otherwise,
    \begin{align}
    \sprob(\tilde{P}^{\density} \leq s) = \prob{P \leq s} + \eta \ell_f\left(s - \int\limits_0^1\prob{P \leq s \mid Q = q}\ dq\right).
\end{align}
\end{proposition}
\begin{proof}
    We only need to prove the last statement, since under independence of $P$ and $Q$, the final statement implies the validity of $\tilde{P}^{\density}$. Let us evaluate the probability $\tilde{P}^{\density}$ lies under some threshold $s \in [0, 1]$:
    \begin{align}
        \sprob(\tilde{P}^{\density} \leq s)
        &= \prob{P \leq s, T = 1} + \prob{Q \leq s, T = 0}\\
        &=\expect[\prob{P \leq s, T = 1 \mid Q} + \prob{Q \leq s, T = 0 \mid Q}]\\
        &=\expect[\prob{P \leq s, T = 1 \mid Q} + \ind{Q \leq s}\prob{T = 0 \mid Q}] \\
        &=\expect[\prob{P \leq s \mid Q} \cdot \prob{T = 1 \mid Q} + \ind{Q \leq s}\prob{T = 0 \mid Q}].\label{eq:reframe}
    \end{align}
    The first and second equalities are purely by reasoning from the law of total expectation and subsequent arithmetic manipulations. The third and fourth equalities are by definition of $T$ being independent of all other randomness when conditioned on $Q$.

    Now, we continue on with our derivations, starting from \eqref{eq:reframe}:
    \begin{align}
         \sprob(\tilde{P}^{\density} \leq s)
         &= \expect\left[\prob{P \leq s \mid Q} \cdot \left(1 - \frac{\eta \ell_f}{f(Q)}\right) + \ind{Q \leq s}\cdot\frac{\eta \ell_f}{f(Q)}\right]\\
         &= \expect\left[\prob{P \leq s \mid Q}\right]  + \eta \ell_f\expect\left[\frac{\ind{Q \leq s} - \prob{P \leq s \mid Q}}{f(Q)}\right]\\
         &= \prob{P \leq s}  + \eta \ell_f\int\limits_0^1\ind{q \leq s} - \prob{P \leq s \mid Q = q}\ dq\\
         &= \prob{P \leq s}  + \eta \ell_f\left(s - \int\limits_0^1\prob{P \leq s \mid Q = q}\ dq\right).
          \label{eq:tight-pval-prob}
    \end{align}

    The first equality is by the construction of the Bernoulli distribution of $T \mid Q$. The third equality is by law of total expectation. Continuing on from \eqref{eq:tight-pval-prob}, we get the following bound:\begin{align}
        \prob{\tilde{P}^{\density}\leq s} &=\prob{P \leq s} +\eta \ell_f\left(s - \int\limits_0^1\prob{P \leq s \mid Q = q}\ dq\right).
    \end{align}
    When $P$ is independent of $Q$ we get that the following is true under the null hypothesis:
    \begin{align}
        \prob{\tilde{P}^{\density}\leq s} = \prob{P \leq s}  + \eta \ell_f(s - \prob{P \leq s}) \leq s,
    \end{align} 
    where the last inequality is due to $\eta \ell_f \leq 1$ and the superuniformity of $P$ under the null hypothesis. We can see that the last inequality is tight, i.e., an equality, when $P$ is exactly uniform. Thus, we have shown our desired result.
\end{proof}

We note that the expected number of queries of the true p-value we make is $\expect[T] = 1 - \eta \ell_f$ 
so we can also adjust the number of times $P$ is computed by varying our choice of $\eta \in [0, 1]$. One drawback of this method is that we are restricted to the minimum budget by the value of $\ell_f$. Our subsequent methods will not have this limitation.

\paragraph{When do $P$ and $Q$ satisfy the assumptions?}When the null hypothesis is simple (i.e., contains only a single possible distribution), it is feasible to either know a priori, or estimate the distribution of $Q$ under the null. However, if one considers a composite null hypothesis (e.g., distribution bounded in $[-1, 1]$ with mean at most 0), there generally is no single distribution that a statistic follows under the null. Moreover, even if such a statistic exists, one cannot guarantee $Q$ is the precise statistic such that it is invariant to which distribution from the null hypothesis is actually the true distribution.

The independence between $P$ and $Q$ can sometimes be achieved by splitting the sample or, if certain parametric distribution assumptions are made about the statistics, one of the data fission/thinning proposals (\citealp{rasines_splitting_strategies_2023}; \citealp{leiner_data_fission_2023}; \citealp{neufeld_data_thinning_2024}; \citealp{dharamshi_generalized_data_2024}) can be applied to derive independent data that can be used to calculate $P$ and $Q$ respectively.
However, data splitting and the aforementioned data fission techniques may also be difficult in cases where the collected data is not i.i.d., e.g., for panel data that is dependent across time. In such a case, it may be more applicable to use the active p-value that only requires validity of the true p-value in \eqref{eq:arbdep-active-pvalue}. We also consider an extension of this active p-value where $P$ and $Q$ are not independent, but we can estimate their joint distribution in \Cref{sec:joint}.
Although theoretical results require independent $P$ and $Q$, synthetic experiments in \Cref{sec:scCRISPR:synthetic} explore active p-values \eqref{eq:density-active-pvalue} under varying correlation levels and show there is only mild type I error inflation under positive dependence.

\section{False discovery rate (FDR) control\label{sec:fdr-control}}

Generalizing from a single hypothesis test, we can also consider the multiple testing setting with $K$ different (null) hypotheses we wish to test. Our goal is to produce a discovery set $\rejset \subseteq [K]$ such that most of the hypotheses in the discovery set are truly non-null. We refer to the ones that are actually null (but are included in $\rejset$) as \emph{false discoveries}. A popular error criterion for multiple hypothesis testing is the \emph{false discovery rate (FDR)}, first introduced by \citet{benjamini_controlling_false_1995}.
The FDR is the expected proportion of discoveries that are false, and is defined as follows:
\begin{align}
    \FDP(\rejset) \coloneqq \frac{|\Ncal \cap \rejset|}{|\rejset| \vee 1}, \qquad \FDR(\rejset) \coloneqq \expect[\FDP(\rejset)],
\end{align} where $\Ncal \subseteq [K]$ denotes the subset of hypotheses that are truly null.
A procedure is considered to have valid FDR control at level $\alpha$ if it always outputs $\rejset$ such that $\FDR(\rejset) \leq \alpha$ for a predetermined level of FDR control $\alpha \in [0, 1]$.
The seminal \emph{Benjamini-Hochberg (BH) procedure} ensures different levels of FDR control under different dependence assumptions on the input p-values. For a p-value vector $\mathbf{P} \coloneqq (P_1, \dots, P_K)$, where $P_{(i)}$ denotes the $i$th smallest p-value in $\mathbf{P}$, the BH procedure is defined as
\begin{gather}
    k^{\rmBH}(\mathbf{P})\coloneqq \max \left\{i \in [K]:P_{(i)} \leq \alpha i / K \right\} \text{ and }\rejset^{\rmBH}(\mathbf{P})\coloneqq \{i \in [K]: P_i \leq \alpha k^{\rmBH}(\mathbf{P}) / K\}.
\end{gather}

\citet{wang_false_discovery_2022} formulated the \emph{e-BH procedure} as a method for producing an FDR controlling set from e-values. Unlike the BH procedure, e-BH is robust to arbitrary dependence among the e-values. Further, \citet{ignatiadis_compound_e-values_2024} showed that \emph{any} FDR controlling procedure is an instance of the e-BH procedure, making e-BH a fundamental procedure for FDR control. Let $\mathbf{E} \coloneqq (E_1, \dots, E_K)$ denote a collection of $K$ e-values, and let $E_{[i]}$ denote the $i$th largest e-value for each $i \in [K]$.
The e-BH procedure applied to a vector of $K$ e-values $\mathbf{E}$, produces the following discovery set:
\begin{align}
    k^{\rmEBH}(\mathbf{E}) \coloneqq \max \left\{i \in [K]:E_{[i]} \geq K / (\alpha i) \right\} \text{ and }\rejset^{\rmEBH}(\mathbf{E}) \coloneqq \{i \in [K]: E_i \geq K / (\alpha k^{\rmEBH}(\mathbf{E}))\}.
\end{align}
We now utilize these procedures in combination with our active framework to derive active multiple testing procedures with FDR control.

\subsection{The active Benjamini-Hochberg (BH) and e-BH procedures}
Since our active framework directly produces valid e-values and p-values, FDR control holds even when using our active e-values with e-BH, or active p-values with BH.

Let $\mathbf{Q} \coloneqq (Q_1, \dots, Q_K)$ and $\mathbf{P} \coloneqq (P_1, \dots, P_K)$ be $K$ proxy p-values and true p-values, respectively, and $\tilde{\mathbf{P}} \coloneqq (\tilde{P}_1, \dots, \tilde{P}_K)$ be the resulting active p-values. Let $\mathbf{F} \coloneqq (F_1, \dots, F_K)$, $\mathbf{E}\coloneqq (E_1, \dots, E_K)$, and $\tilde{\mathbf{E}}\coloneqq (\tilde{E}_1, \dots, \tilde{E}_K)$ denote the proxy, true, and active e-values for each hypothesis. We then define the \emph{active BH} and \emph{active e-BH} procedures as follows:
\begin{align}
\rejset^\actBH \coloneqq \rejset^\BH(\tilde{\mathbf{P}}), \qquad
\rejset^\acteBH \coloneqq \rejset^\eBH(\tilde{\mathbf{E}}).
\end{align}

We formulate the complete procedures in \Cref{alg:active-bhs}.
\begin{algorithm}[h]\label{alg:active-bhs}
\caption{Active e-BH and BH using active e-values and p-values.}
\KwIn{Desired level of FDR control $\alpha \in [0, 1]$. Tuning parameter $\gamma \in [0, 1]$. Proxy e-values $(F_1, \dots, F_K)$ or proxy p-values $(Q_1, \dots, Q_K)$.}
\begin{minipage}[t]{0.45\textwidth}
\tcbox[hbox,nobeforeafter,tcbox raise base]{active BH}
\If{using p-values }{
\For{$i \in \{1, \dots, K\}$}{
    Sample $T_i \sim \text{Bern}(1 - \gamma Q_i)$.\\
    \eIf{$T_i = 1$}{
        Query true p-value $P_i$.\\
        Compute $\tilde{P}_i = (1 - \gamma)^{-1} \cdot P_i$
    }{
        Set $\tilde{P}_i = Q_i$
    }
}
Output $\rejset^\actBH \coloneqq \rejset^\BH(\tilde{\mathbf{P}})$
}
\end{minipage}\begin{minipage}[t]{0.45\textwidth}
\tcbox[hbox,nobeforeafter,tcbox raise base]{active e-BH}
\If{using e-values }{
\For{$i \in \{1, \dots, K\}$}{
    Sample $T_i\sim\text{Bern}(1 - \gamma F_i^{-1})_+$.\\
    \eIf{$T_i = 1$}{
        Query true e-value $E_i$.\\
        Compute $\tilde{E}_i = (1 - \gamma) \cdot E_i$
    }{
        Set $\tilde{E}_i = F_i$
    }
}
Output $\rejset^\acteBH \coloneqq \rejset^\eBH(\tilde{\mathbf{E}})$}
\end{minipage}
\end{algorithm}
We can formalize the FDR guarantee of active e-BH in the following theorem.
\begin{theorem}\label{thm:active-bh}
    Active BH ensures FDR control for $\rejset^\actBH$ for different dependence structures in $\mathbf{P}$.
    \begin{align}
    \FDR(\rejset^\actBH)\leq
        \begin{cases}
            \alpha(1 + \log(\alpha^{-1})) &
            \begin{aligned}
             &   \text{independent or }\text{positively dependent}\\
            &\text{(PRDN;  \Cref{sec:prds-def})}
            \end{aligned}\\
            \alpha(3.18 + \log(\alpha^{-1})) & \text{negatively dependent (WNDN; \Cref{sec:prds-def})}\\
            \alpha \ell_K & \text{arbitrarily dependent}
        \end{cases}.
    \end{align}
\end{theorem}
We defer the proof to \Cref{sec:active-bh-proof}.
It remains to be seen if FDR control in the independent or positively dependent case can be improved to $\alpha$, and we leave that to future work.

\begin{theorem}
    The active e-BH procedure ensures $\FDR(\rejset^{\acteBH}) \leq \alpha$ under arbitrary dependence among $\mathbf{F}$ and $\mathbf{E}$.
\end{theorem}
This follows simply from $\mathbf{\tilde{E}}$ being a vector of e-values and FDR control of the e-BH procedure from Theorem 2 of \citet{wang_false_discovery_2022}.
Thus, both BH and e-BH provide FDR control in the active setting.
We will consider extensions of these procedures with FDR control where the proxy statistics of all hypotheses are computed or altered jointly or updated iteratively after each new statistic is queried.

Note here for the $i$th hypothesis, arbitrary dependence between its proxy p-value $Q_i$ and the true p-value $P_i$ is accounted for by the design of the active p-value.
Hence, the change in FDR control is only affected by the dependence between the active p-values of different hypotheses.

\paragraph{Tuning $\gamma$ and $\eta$}\label{sec:gamma-tune} We can set up the problem of optimizing $\gamma \in [0, 1]$ as finding the optimal value to maximize growth rate of the active e-value, as is commonly done in the e-value literature for maximizing power \citep{grunwald_safe_testing_2020,ramdas_gametheoretic_statistics_2022}. In the simplest setting, we assume that our $(X, Z)$ are i.i.d.\ draws from a two groups model, i.e., we sample $(X, Z) \sim \pdist$ with probability $\pi_0 \in [0, 1]$ and an alternative distribution $(X, Z) \sim Q$ with probability $1 - \pi_0$. Let this mixture distribution be defined as $\bar{\pdist} = \pi_0 \pdist + (1 - \pi_0) Q$.
\begin{align}
    &\max_{\gamma \in [0, 1]}\ \expect[\log(\tilde{E})]\\
     = &\max_{\gamma \in [0, 1]}\ \expect[(\gamma F^{-1} \wedge 1)\log(F) + (1 - \gamma F^{-1})_+(\log \gamma E)]. 
\end{align} 
If we want to impose a sampling budget constraint $\expect[T] \le B$, we can then impose the equivalent condition that
\begin{align}
    \expect\left[ \left(1 - \frac{\gamma}{F}\right)_+ \right] \leq B.
\end{align}
Although the optimization objective is non-smooth at $\gamma = F$, we can use grid search over $[0, 1]$ and the empirical approximation of the expectation under $\bar{\pdist}$ to derive the choice of $\gamma$ in practice. We note that this optimization formulation also provides a principle for tuning $\gamma$ for active p-values, since we can directly substitute $F = Q^{-1}$ and $E = P^{-1}$ in the above objective and constraint as well.

We may also vary the hyperparameter $\eta$ for the $\tilde{P}^{\textnormal{density}}$ active p-value. 
Because the expected number of queries of the true p-value is $\expect[T] = 1 - \eta \ell_f$, the hyperparameter $\eta$ may be changed to increase or decrease the expected total computation cost. For example, setting $\eta=1$ minimizes the expected computation cost while still ensuring that $\tilde{P}^{\textnormal{density}}$ is a valid p-value, but in some cases, one may be willing to spend a higher computation cost in order to query the true p-value more often. This can be done by setting smaller $\eta$ values.

\subsection{Proxy-filter (PF) and e-PF: filtering hypotheses directly using proxies}
\label{sec:self-consistent}

In the multiple testing setting, applying FDR controlling procedures to active statistics is not the sole approach for selectively querying true statistics.
For both p-values and e-values, we can use the proxy statistics to filter out a subset of hypotheses for which we will query the true statistic, even when the proxies and true statistics are arbitrarily dependent.
In the following procedure, we will \emph{deterministically} choose which true statistics to query based on the proxies.
Define a \emph{proxy selection algorithm} $\selalg_P: [0, 1]^K \rightarrow 2^{[K]}$ (for p-values) or $\selalg_E: [0, \infty)^K \rightarrow 2^{[K]}$ (for e-values). These are functions chosen by the user to select a subset of hypotheses to query based on the proxies.

\begin{algorithm}[h]\label{alg:sc-filter}
    \caption{Proxy-Filter applies BH to p-values (and e-Proxy-Filter applies e-BH to e-values) after querying a deterministic selection of true statistics based on the proxies.}
    \KwIn{Desired FDR control $\alpha\ \in [0, 1]$. Proxy p-values $(Q_1, \dots, Q_K)$ or proxy e-values $(F_1, \dots, F_K)$.
    Proxy selection algorithm $\mathfrak{S}_P: [0, 1]^K \rightarrow 2^{[K]}$ (for p-values) or $\mathfrak{S}_E: [0, \infty)^K \rightarrow 2^{[K]}$ (for e-values).
}
    \resizebox{0.9\textwidth}{!}{ \begin{minipage}[t]{0.48\textwidth} \tcbox[hbox,nobeforeafter,tcbox raise base]{Proxy-Filter (PF)}
    \If{using p-values }{
    $\mathcal{S}=\mathfrak{S}_P(Q_1, \dots, Q_K)$.\\
    For $ i \in \mathcal{S}$, obtain $P_i$, set $\tilde{P}_i = P_i$. \\
    Set $\tilde{P}_i = 1$ for $i \not\in \mathcal{S}$.\\
    Reject $\Rcal^{\pf} \coloneqq \Rcal^{\text{BH}}(\tilde{\mathbf{P}})$.
    }
    \end{minipage}
    \hfill
    \begin{minipage}[t]{0.48\textwidth} \tcbox[hbox,nobeforeafter,tcbox raise base]{e-Proxy-Filter (e-PF)}
    \If{using e-values }{
        $\mathcal{S}=\mathfrak{S}_E(F_1, \dots, F_K)$ .\\
        For $ i \in \mathcal{S}$, obtain $E_i$ and set $\tilde{E}_i = E_i$.\\
        Set $\tilde{E}_i = 0$ for $i \not\in \mathcal{S}$.\\
        Reject $\Rcal^{\epf} \coloneqq \Rcal^{\rmEBH}(\tilde{\mathbf{E}})$.
    }
    \end{minipage}}
\end{algorithm}
Now, we note the following class of discovery sets \citep{blanchard_two_simple_2008,su_fdr_linking_2018}.
\begin{definition}\label{def:sc}
    A discovery set $\Rcal \subseteq [K]$ is called \emph{self-consistent} if $P_i \leq \alpha|\Rcal| / K$ (if using p-values) or $E_i \geq (\alpha|\Rcal|)^{-1}K$ (if using e-values) for each $i \in \Rcal$.
\end{definition}
Now, let $\selset_P \coloneqq \selalg_P(\mathbf{P})$, and $\selset_E \coloneqq \mathfrak{S}_E(\mathbf{E})$ be the selected sets.
\begin{align}
    \rejset^\pf &\coloneqq \rejset^\BH((P_1 \vee \ind{1\not\in \selset_P}, \dots, P_K \vee \ind{K\not\in \selset_P}))\\
    \rejset^\epf &\coloneqq \rejset^\eBH((E_1\cdot \ind{1\in \selset_E}, \dots, E_K \cdot \ind{K\in \selset_P})),
\end{align}
Essentially, we are applying BH and e-BH to the selected set of p-values and e-values, with the unselected statistics being set to trivial values (1 for p-values, and 0 for e-values). The full procedure is formulated in \Cref{alg:sc-filter}.
\begin{proposition}
    For a vector of p-values $\mathbf{P}$, $\Rcal^\pf$ outputs a self-consistent discovery set. Similarly, for a vector of e-values $\mathbf{E}$, $\Rcal^\epf$ outputs a self-consistent discovery set as well.
\end{proposition}
\begin{proof}
    BH and e-BH always output a self-consistent discovery set w.r.t.\ their input statistics by definition, and the selection set $\mathcal{S}$ simply enforces all hypotheses outside of it will never be selected by setting their p-values or e-values to be 1 or 0, respectively. Since the $\tilde{P}_i = P_i$ and $\tilde{E}_i = E_i$ for each $i \in \mathcal{S}$, applying BH and e-BH, respectively, will output self-consistent discovery sets relative to $\mathbf{P}$ and $\mathbf{E}$, respectively.
\end{proof}
\begin{theorem}
    $\FDR(\Rcal^\epf(\mathbf{E})) \leq \alpha$ for arbitrarily dependent e-values $\mathbf{E}$. We have the following bounds on the FDR for the proxy-filter discovery set when using p-values $\mathbf{P}$:
    \begin{align}
        \FDR(\Rcal^\pf(\mathbf{P}))\leq
        \begin{cases}
            \alpha(1 + \log(\alpha^{-1})) &
            \text{independent or positively dependent (\Cref{sec:prds-def})}\\
            \alpha(3.18 + \log(\alpha^{-1})) & \text{WNDN; negatively dependent (\Cref{sec:prds-def})}\\
            \alpha \ell_K & \text{arbitrarily dependent}
        \end{cases}.
    \end{align}
\end{theorem}
\begin{proof}
    Proof of the FDR bound of e-BH comes from FDR control of self-consistent procedures in Theorem 2 of \citet{wang_false_discovery_2022}. FDR control in the independent, PRDS, and arbitrary dependent cases for p-values comes from the FDR bound for self-consistent p-value procedures in \citet{su_fdr_linking_2018}. The remaining negative dependence FDR bound comes from Proposition 3.6 of \citet{fischer_online_generalization_2024}.
\end{proof}

For p-values, proxy filtering results in inflation of the level of guaranteed FDR control, even under independence. For e-values, there is no inflation of FDR control with the proxy filter procedure; however, unlike with the other active multiple testing procedures, proxy filtering does not use the proxy statistic for calculating the final statistic.

\section{Numerical simulations}\label{sec:simulations}

We conduct numerical simulations to demonstrate that the active statistics we have defined provide power while saving significantly in the number of queries.

\begin{figure}[ht]
    \centering
    \begin{subfigure}[t]{0.11\textwidth}
        \includegraphics[width=\textwidth,trim=0 0 2in 0, clip]{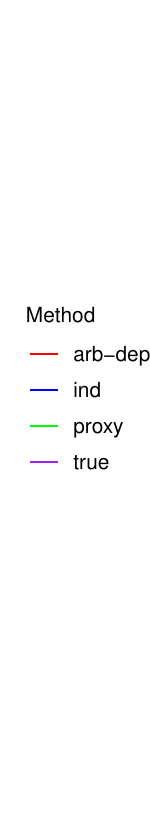}
    \end{subfigure}
    \begin{subfigure}[t]{0.28\textwidth}
        \centering
        \includegraphics[width=0.8\textwidth,trim=1in 0 0 0, clip]{figures/pval_corr/cor_plot_cdf_null.pdf}
        \caption{True null ($\mu = 0$).}
    \end{subfigure}\hfill
    \begin{subfigure}[t]{0.28\textwidth}
        \centering
        \includegraphics[width=0.8\textwidth,trim=1in 0 0 0, clip]{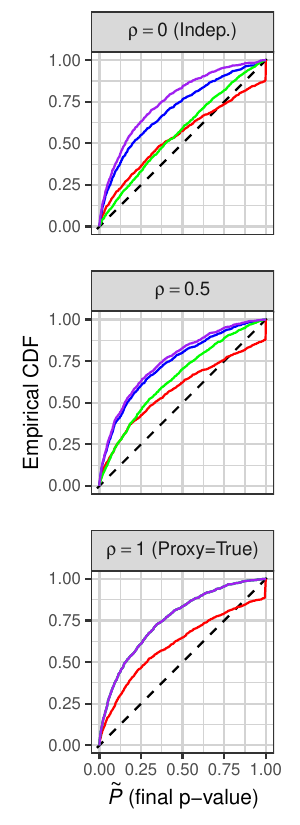}
        \caption{False null ($\mu > 0$).}
    \end{subfigure}\hfill
    \begin{subfigure}[t]{0.28\textwidth}
        \centering
        \includegraphics[width=0.835\textwidth,trim=1.2in 0 0 0, clip]{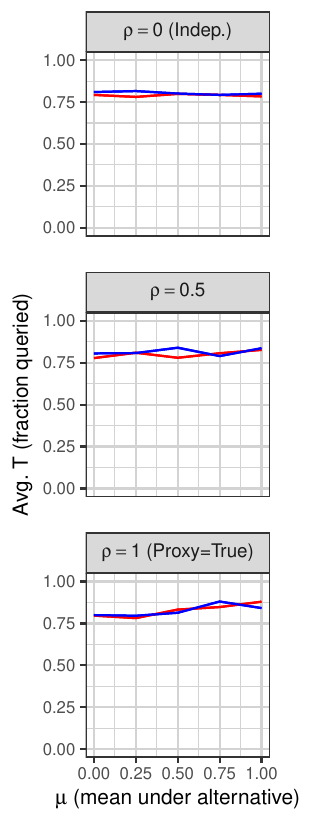}
        \caption{Avg. number of queries.}
    \end{subfigure}
    \caption{Plots of the p-values' empirical c.d.f. and their query frequencies for different active p-values.  Black dotted line marks c.d.f. of the uniform distribution. The proxy p-value is invalid (above uniform) while both ``ind'' and ``arb-dep'' active p-values are valid (below uniform) under the null. Both active p-values are powerful when the null is false and have similar average queries across different values of $\mu$.}
    \label{fig:pvalue-corr-gaussian-sim}
\end{figure}

We will call the active p-value in \eqref{eq:density-active-pvalue} the ``ind'' active p-value, and the active p-value formulated in \eqref{eq:arbdep-active-pvalue} as the ``arb-dep'' active p-value.
In this setting, $Z$ and $X$ are jointly a multivariate Gaussian with positive correlation $\rho$. We let $\expect[Z] = 0$ and $\expect[X] = \mu_{X, 0} = 0.3$ under the null, and $\expect[Z] = \rho \mu$ and $\expect[X] = \mu$ under the alternative. We choose different values of $\mu > 0$ under the alternative.
We define our proxy and true p-values as follows:
\begin{align}
   Q \coloneqq 1 - \Phi(X), \qquad  P  \coloneqq 1 - \Phi(Z)
\end{align} where $\Phi$ is the c.d.f. of the standard normal distribution. As a result, we note that the density of $Q$ under the null hypothesis is the following:
\begin{align}
    f(q) \coloneqq \frac{\varphi(\Phi^{-1}(1 - q))}{\varphi(\Phi^{-1}(1 - q) - \mu_{X, 0})},
\end{align} where $\varphi$ is the p.d.f.\ of the standard normal and $\mu_{X, 0}$ is the mean of $X$ when the null is true, i.e., $\expect[Z] = 0$. We numerically solve for a lower bound on $f(q)$ and find that it is approximately $\eta \ell_f = 0.1876$. We average our following results from 1000 trials.

We see in \Cref{fig:pvalue-corr-gaussian-sim} that both active p-value methods retain a superuniform distribution under the null, although the ``arb-dep'' method is more conservative. We also see that both active p-value methods have power under the alternative distribution where $\mu = 1$. Further, we see that both methods use a similar frequency of queries on average, even across different choices of $\mu$.

\section{Testing for causal effects in scCRISPR screen experiments}
\label{sec:scCRISPR}

\begin{figure}[ht]
    \centering
    \includegraphics[width=.75\textwidth]{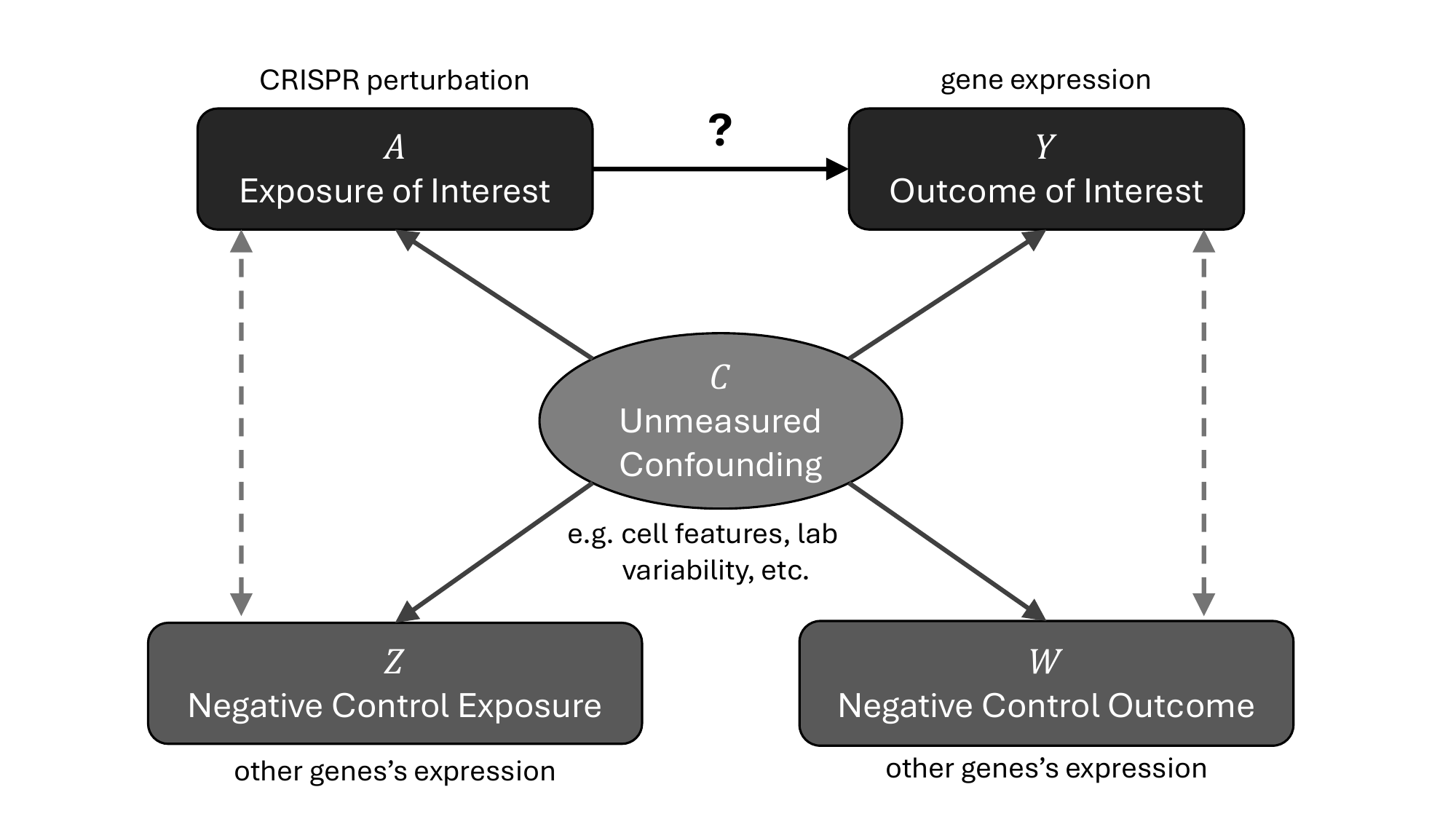}
    \caption{A directed acyclic graph that illustrates the causal relationships assumed in the proximal causal inference framework. In particular, it highlights the relationship between the negative control exposures and negative control outcomes with the exposure and outcome. Solid arrows represent causal relationships, and their absence indicates that there is no causal relationship. Dashed arrows represent allowed causal relationships. The question mark highlights the causal relationship we wish to test.
    \label{fig:proxDAG}
    }
\end{figure}

CRISPR is a tool that allows scientists to target and perturb specific regions of the DNA, and single-cell RNA-seq is a tool that allows scientists to measure the expression levels of genes in individual cells.
Together, single-cell CRISPR (scCRISPR) screen or Perturb-seq (where the gene perturbations are induced) experiments can be performed to identify and measure how one region of the DNA may regulate the expression of a full panel of genes (\citealp{Kampmann:2020}; \citealp{Hong:2023}; \citealp{Cheng:2023}).
Thus, in scCRISPR screens, scientists wish to test for the causal effect of a CRISPR perturbation on a gene's expression. However, there is often unmeasured confounding that biases the results, e.g., varying laboratory conditions, different cell characteristics such as cell phase or cell size, etc. Unmeasured confounding renders naive estimates of the causal effect (e.g., the coefficient of ordinary least squares) invalid as a result.

We propose accounting for unmeasured confounding via proximal causal inference \citep{miao_confounding_2020, cui_semiparametric_2023}.
However, computation of the unbiased proximal causal statistic is much more expensive than running a fast but biased least squares test.
In terms of sample size $n$ and number of negative controls $d$, computing the proximal causal statistic has a runtime complexity of $O(nd^4 + d^6)$ compared to a least squares test that only requires $O(n)$ (see \Cref{sec:appendix:scCRISPR:2SLSdetails}).
Empirically, we observe that computing each proximal causal statistic takes about 34 seconds, whereas the fast but biased test requires only .0046 seconds.
Because we wish to use proximal causal inference on hundreds of thousands of hypotheses (i.e., each perturbation paired with each gene), computing the proximal causal statistics can be quite computationally expensive.
Thus, we can use active statistics to avoid computing the full proximal causal p-value for each test, for which we expect the vast majority to be truly null, while still retaining power and discovery of most of the important causal relationships between a perturbation and gene expression.

\subsection{Overview of proximal causal inference for scCRISPR screens}\label{sec:scCRISPR:intro}
Proximal causal inference assumes that one has access to other random variables that are also affected by the unmeasured confounding.
These random variables are called \emph{negative controls} and we use them to produce unbiased causal effect estimates and valid p-values.
\Cref{fig:proxDAG} depicts a causal directed acyclic graph (DAG) that satisfies the proximal causal inference setup.

Let $A \in \{0, 1\}$ be an indicator random variable of whether we applied a perturbation to a gene, and $Y \in \reals$ denote the expression of that gene.
The goal is to test if $A$ has a causal effect on $Y$.
We assume we have access to negative control exposures (NCEs), $Z$, and negative control outcomes (NCOs), $W$.
Both types of negative controls are causally affected by the unmeasured confounding $C$, but the NCEs can only be causally related to $A$ while the NCOs can only be causally related to $Y$.
Our goal is to test if the average treatment effect (ATE) is nonzero, which is defined as $\psi\coloneqq \expect[Y^1] -\expect[Y^0]$. Here, $Y^a$ is the potential outcome if we set the treatment to level $a$, for $a \in \{0,1\}$ --- $Y^1$ is the cell's gene expression if the perturbation were to be applied, and $Y^0$ if it were not. Formally, we test the hypothesis
\begin{align}
    H_0  : \psi = 0\text{ vs. }H_1  : \psi \neq 0.
\end{align}
We can construct p-values by calculating an estimator of the ATE, $\hat\psi$, and deriving its asymptotic distribution.

Under the standard assumptions of \emph{consistency} ($Y = Y^a$ when $A=a$), \emph{positivity} (there exists $\varepsilon > 0$ s.t.\ $\prob{A=a|C} > \varepsilon$ for all $a \in \{0, 1\}$), and \emph{no unmeasured confounding} (all confounders $C$ are measured and $Y^1, Y^0 \indep A \mid C$),
the mean potential outcomes can be identified with
\begin{align}
\expect[Y^a]
= \expect[\expect[Y^a| C]]
= \expect[\expect[Y^a|A = a,C]]
= \expect[\expect[Y|A = a,C]]
\end{align}

However, when we cannot measure the confounders $C$, $\expect[Y \mid A = a, C]$ is no longer identifiable. Instead, in proximal causal inference, we make the following main additional assumptions, taken from \citet{miao_confounding_2020}.
\begin{assumption}[Proximal causal inference assumptions]
    We make the following assumptions about the $Y, A, C, W$, and $Z$.
\begin{enumerate}
    \item (negative control outcome) $W \indep A \mid C$ and $W \not\indep C$ .
    \item (negative control exposure) $Z \indep Y \mid C, A$ and $Z \indep W \mid C$.
\end{enumerate}
\end{assumption}
The comprehensive set of assumptions we make to conduct proximal causal inference is enumerated in \Cref{sec:appendix:scCRISPR:proximalassumptions}.

At a high level, proximal causal inference uses the existence of these negative controls to remove the influence of confounders. However, a major challenge in applying proximal inference is forming both valid and strong negative control exposures and outcomes.

Under a linear setting, \cite{miao_confounding_2020} and \cite{tchetgen_introduction_proximal_2020} show that the ATE may be estimated by a two-stage least squares (2SLS) procedure.  In the first stage, the negative control outcome $W$ is regressed onto the treatment $A$ and negative control exposures $Z$.
Then, in the second stage, the outcome $Y$ is regressed onto the treatment $A$ and the estimated values of $W$ from the first stage.
The estimated coefficient of $A$ in the second stage is the estimate of the ATE.
We use the implementation by \citep{liu_regression-based_2024} which also provides statistical results including an asymptotic distribution and p-value.

The biased but cheap p-value is the p-value from a naive linear regression comparing the gene expression $Y$ to perturbation assignment $A$.
The desirable unbiased but expensive p-value is the p-value from proximal causal inference using the outcome regression model for proximal inference implemented by \cite{liu_regression-based_2024}.
For sample size $n$, let the response vector of gene expression be $\vv{\mathbf{Y}} \in \reals^n$, the perturbation assignment be $\vv{\mathbf{A}} \in \{0, 1\}^n$, the negative control exposures be $\mathbf{Z} \in \reals^{n\times d}$, and the negative control outcomes be $\mathbf{W} \in \reals^{n\times d}$.

We formulate how we compute the proxy and true p-value in \Cref{alg:2sls}. $F_{\text{df}}$ denotes the c.d.f.\ of the t-distribution with df degrees of freedom and $\Phi$ is the c.d.f.\ of the standard normal distribution.
\Cref{sec:appendix:scCRISPR} provides more details about how the standard errors are derived and estimated, and about the runtime complexity of the true and biased tests.
\begin{algorithm}[h]\label{alg:2sls}
    \caption{Algorithm for calculating proxy and true p-values using 2SLS}
\SetAlgoLined
\KwData{Gene expression vector $\vv{\mathbf{Y}} \in \reals^n$, perturbation assignment $\vv{\mathbf{A}} \in \{0, 1\}^n$, negative control exposures $\mathbf{Z} \in \reals^{n\times d}$, negative control outcomes $\mathbf{W} \in \reals^{n\times d}$.}
\KwOut{Active p-value $\tilde{P}$.}
Compute
$\hat\beta^\OLS \coloneqq (\mathbf{M}^T \mathbf{M})^{-1} \mathbf{M}^T \vv{\mathbf{Y}}\text{ where } \mathbf{M} = [\mathbf{1}, \vv{\mathbf{A}}]$
(naive OLS).\\
Compute
$\hat{\sigma}^\OLS \coloneqq \sqrt{\left((\mathbf{M}^T \mathbf{M})^{-1}  \hat\sigma^2\right)_{A,A}}$ where $\hat\sigma = \sqrt{\sum_i^n (Y - \mathbf{M}\beta)^2/(n-2)}$.\\
$Q = 2 \cdot F_{n-2}(-|\hat\psi^\OLS|/ \hat{\sigma}^\OLS) \text{ where } \hat\psi^\OLS \coloneqq \hat\beta^\OLS_A$.\\
Sample $T \sim \text{Bern}(1 - \gamma Q)$.\\
\eIf{$T = 1$}{
Compute $\mathbf{S} \coloneqq (\tilde{\mathbf{Z}}^T \tilde{\mathbf{Z}})^{-1} \tilde{\mathbf{Z}}^T \mathbf{W} \text{ where } \tilde{\mathbf{Z}} = [\mathbf{1}, \vv{\mathbf{A}}, \mathbf{Z}]$. (first stage OLS)\\
Compute $\hat\beta^\TSLS \coloneqq (\tilde{\mathbf{S}}^T \tilde{\mathbf{S}})^{-1} \tilde{\mathbf{S}}^T \vv{\mathbf{Y}}\text{ where } \tilde{\mathbf{S}} = [\mathbf{1}, \vv{\mathbf{A}}, \mathbf{S}]$. (second stage OLS)\\
Compute $\hat\sigma^\TSLS$ (\underline{computationally expensive step} --- see \Cref{sec:appendix:scCRISPR}).\\
$P = 2 \cdot \Phi(-|\hat\psi^\TSLS|/ \hat{\sigma}^\TSLS) \text{ where } \hat\psi^\TSLS \coloneqq \hat\beta^\TSLS_A \label{eq:2sls-pvalue}$.\\
Set $\tilde{P} = P$.}
{Set $\tilde{P} = Q$.}
\end{algorithm}
The true p-value we use in \Cref{alg:2sls} and our resulting active p-value is valid asymptotically (in the sense of sample size), as stated in the following result.
\begin{proposition}
    Under a linear structural equation model (eqs. (4) and (5) of \cite{liu_regression-based_2024}, which satisfy the proximal causal inference assumptions \Cref{sec:appendix:scCRISPR:proximalassumptions}) $\expect[Y|A, Z, U] = \beta_0 + \beta_a A + \beta_u U$ and $\expect[W|A,Z,U] = \alpha_0 + \alpha_u U$, the p-value $P$ calculated by \Cref{alg:2sls} is asymptotically a valid p-value.
\end{proposition}
\begin{proof}
    Appendix A.9 of \cite{liu_regression-based_2024} shows that the test statistic $\hat\psi^\TSLS / \hat{\sigma}^\TSLS$ has an asymptotic standard normal distribution, which means that the true p-value $P$ in \Cref{alg:2sls} converges in distribution to a Uniform$[0,1]$ and is asymptotically valid.
\end{proof}

\subsection{Real data from an scCRISPR screen experiment}\label{sec:scCRISPR:res}

\begin{figure}[ht]
    \centering
    \includegraphics[height=4.8cm, trim=0 0 1.5in .3in, clip]{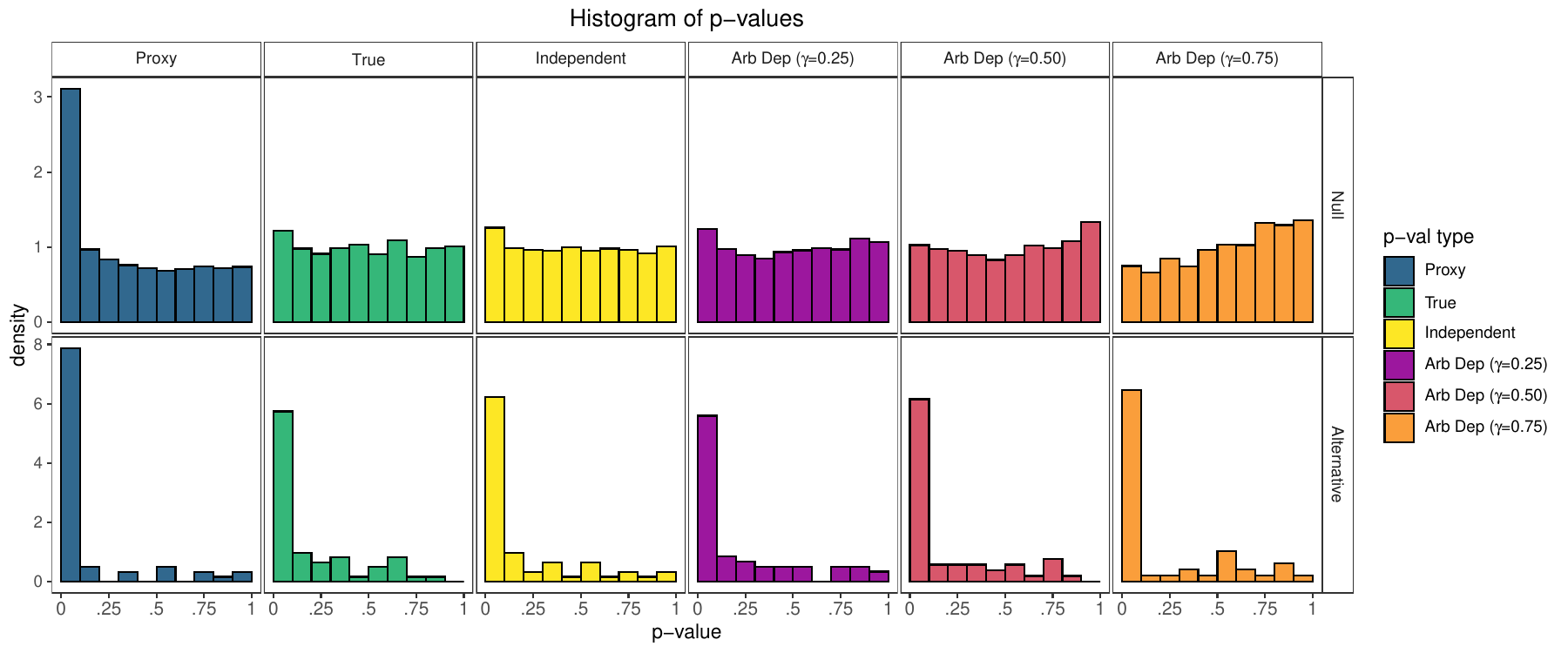}\hfill \includegraphics[height=4.5cm, trim=0 0 1.5in .3in, clip]{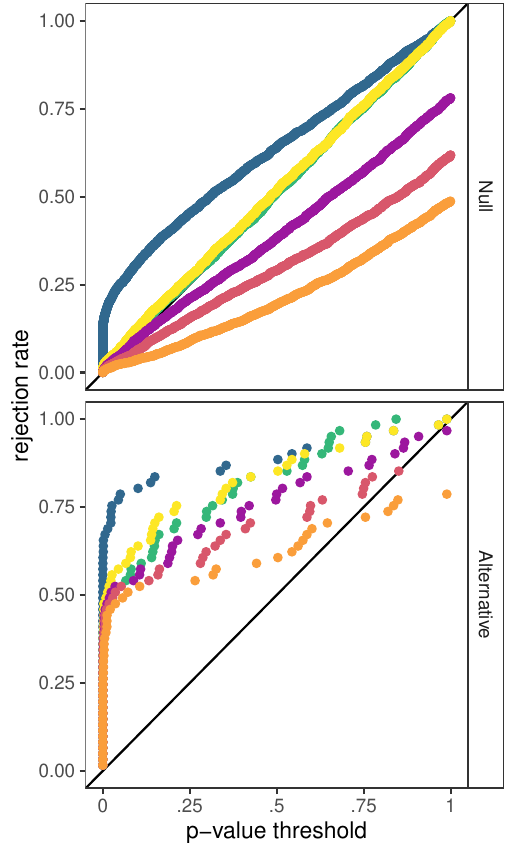}
    \includegraphics[height=4.2cm, trim=10.5in 0in 0  1.2in, clip]{figures/papalexi/all_hist01.pdf}
    \caption{Histogram and rejection rate plots of p-values from the scCRISPR screen experiment in \citet{Papalexi2021}. ``Proxy'' refers to the naive comparison of gene expression without accounting for confounding. ``True'' refers to the proximal causal inference method accounting for unmeasured confounding. We consider the independent with known density active p-value \eqref{eq:density-active-pvalue} and the arbitrary dependence active p-values with varying levels of $\gamma$ \eqref{eq:arbdep-active-pvalue}. The top row shows results for mostly null hypothesis tests of interest, and the bottom row shows results under the alternative.
    }
    \label{fig:papalexi-results}
\end{figure}

\begin{figure}[ht]
    \centering
    \begin{subfigure}[t]{.4\textwidth}
        \includegraphics[height=6cm, center, trim=0 0 1.25in 0, clip]{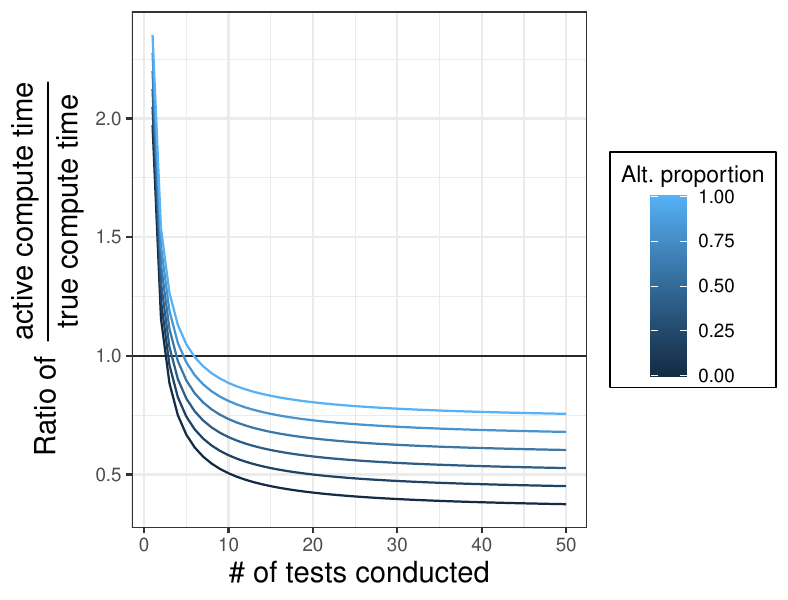}
    \end{subfigure}\begin{subfigure}[t]{.4\textwidth}
        \includegraphics[height=6cm, center, trim=0 0 1.25in 0, clip]{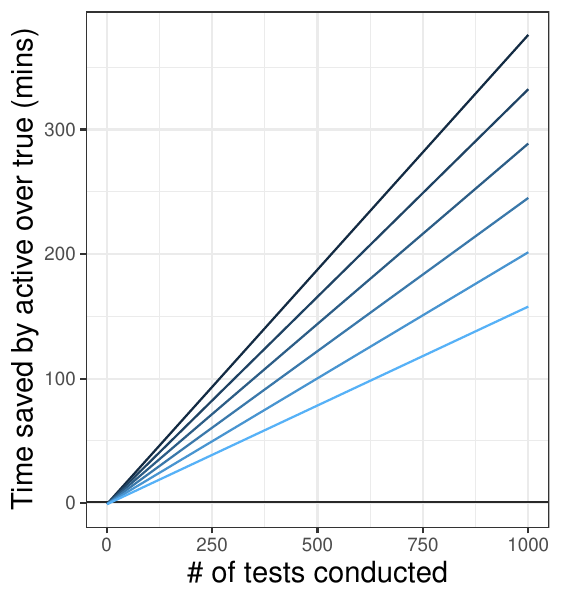}
    \end{subfigure}
    \begin{subfigure}[t]{.1\textwidth}
        \includegraphics[height=6cm, left, trim=3.75in 0 0in 0, clip]{figures/papalexi/comp_time_diff.pdf}
    \end{subfigure}
    \caption{Plots of the reduction in computation time with the active p-value method in \eqref{eq:density-active-pvalue} compared to computation time when using the true, expensive procedure on all tests. We estimate computation time using the observed value in the scCRISPR screen analysis.
    The true p-value was queried in 34\% of the null hypothesis tests and 72\% of the alternative hypothesis tests, resulting in a large decrease in computation time as more hypotheses are tested.
}
    \label{fig:papalexi-comptime}
\end{figure}

We apply our methods to data from a controlled experiment performed by \cite{Papalexi2021} where there are around 20,000 genes measured across approximately 18,000 cells. Each cell receives one of 110 perturbations, 101 of which are targeting perturbations and 9 of which are non-targeting. The targeting perturbations should cause an effect on a subset of the measured genes, including their targeted genes, and the non-targeting perturbations should not cause an effect on any gene.

For a particular perturbation-gene test, we use the cells that received this perturbation as treatment cells, $A=1$, and
we use the cells that received any non-targeting perturbation as control cells, $A=0$.
We restrict our investigation to the top 2000 ``most important genes", which are determined by \cite{townes_feature_2019}; these are often called highly variable genes. We also only investigate perturbations that were applied to $200$ or more cells.
For a perturbation applied to a given gene, we anticipate an effect in the response of that gene. Occasionally, when the perturbation is ineffective, these tests will be null, but a majority should be alternative, and we will assume that the alternative is true for all genes where a perturbation is applied. These resulting $61$ featured tests are considered to be ``alternative" tests.
On the other hand, the vast majority of genes will be unaffected by perturbations, and we will assume that the outcomes of these tests are drawn from a null distribution (see \Cref{fig:papalexi-results}). There are hundreds of thousands of these ``null'' tests, and for simplicity we randomly choose 4300 of them to visualize.

We consider the p-values resulting from a naive comparison of gene expression values as proxy p-values, and we consider the p-values from proximal inference, which should account for unmeasured confounding, as true p-values.
To derive $f$, the density of $Q$, we learn a local polynomial density estimator \citep{cattaneo_2020}, $\hat{f}$.
This is trained on $2000$ randomly chosen realizations of $Q$ that correspond to perturbation-gene tests that are assumed to be null,
which leaves $2300$ tests to evaluate the active procedure.
We directly use $f = \hat{f}$ as the density of $Q$ in our active p-value in \eqref{eq:density-active-pvalue}. Further, we use the lower bound on our estimated density $\ell_{f} = \ell_{\hat{f}}$ and choose $\eta = 1$.

\paragraph{Results} \Cref{fig:papalexi-results} shows the results of the data analysis on the scCRISPR data.
Of the 61 alternative hypothesis tests, the active p-values retain good power.
Of the remaining 2300 null hypothesis tests, the biased proxy p-values are generally stochastically smaller than uniform. Empirically, we can see that they are invalid p-values. By contrast, the active p-value method produces p-values that are relatively uniform, and hence seem valid.
Further, this is evidence that the resulting active p-values can be robust to potential violations of the assumptions of proximal causal inference. The original scCRISPR experiment may include gene to gene or perturbation to gene causal pathways that violate the assumptions of our negative controls.
Second, we violate the active p-value assumption that the density of the null hypothesis tests may be known or estimated. In this case, we have a set of hypothesis tests that we deem ``null'' even though they may be contaminated with non-null hypotheses.
Third, the validity guarantee in \Cref{prop:density-act-pvalue-valid} requires that the proxy and true p-value are independent.
However, we see that the resulting active p-values for known null tests seem valid despite a small but nonzero empirical Spearman rank correlation of $0.25$.

Additionally, we investigate the reduction in computation time. \Cref{fig:papalexi-comptime} shows the ratio of active computation time to the standard true computation time (using the true test on all tests) across different proportions of null and alternative tests. The reduction in computation time may vary depending on the setting, such as the proxy time, the true time, the proportion of hypotheses where true p-values were queried, etc. We calculated the computation time based on empirical time benchmarks in this scCRISPR screen experiment analysis. Initially, the active procedure includes a fixed cost from the estimation of some null proxies and their density. As the number of hypothesis tests increases, the reduction of computation time converges to a proportion that is dependent on the distribution of null proxy p-values, proportion of alternative tests, and difference between proxy and true p-value computation cost.

\subsection{Synthetic experiments via rank-correlated p-values}\label{sec:scCRISPR:synthetic}

\begin{figure}[ht]
    \centering
    \includegraphics[height=7cm, center]{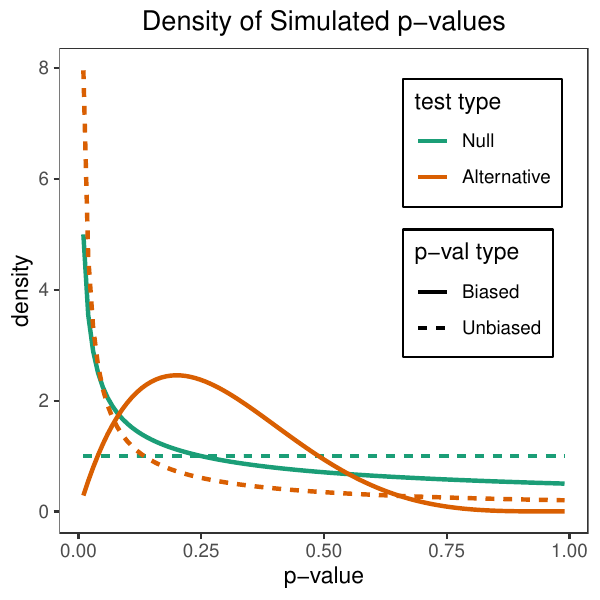}
    \caption{Density of $Q$ and $P$ for semi-synthetic simulations that approximates the empirical distribution of experimental data in \citet{Papalexi2021}.}
    \label{fig:simBeta_setup}
\end{figure}

\begin{figure}[ht]
    \centering
    \begin{subfigure}[t]{.45\textwidth}
        \centering
\includegraphics[height=6cm, center, trim=0 0 0 .31in, clip]{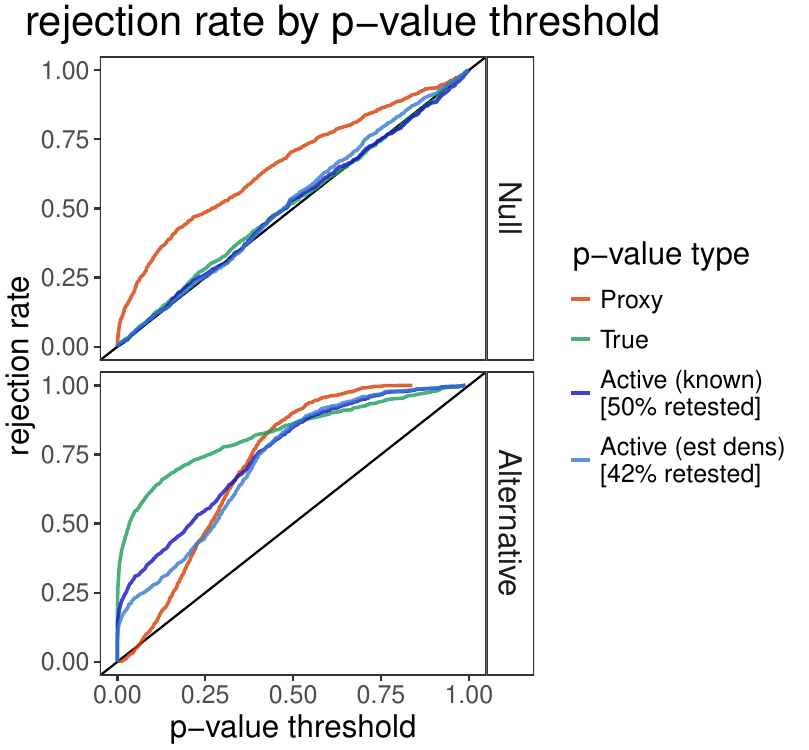}
        \caption{Independent $Q$ and $P$}
        \label{fig:simBeta_qqplot}
    \end{subfigure}\hfill\begin{subfigure}[t]{.45\textwidth}
        \centering
\includegraphics[height=6cm, left, trim=0 0 0 .31in, clip]{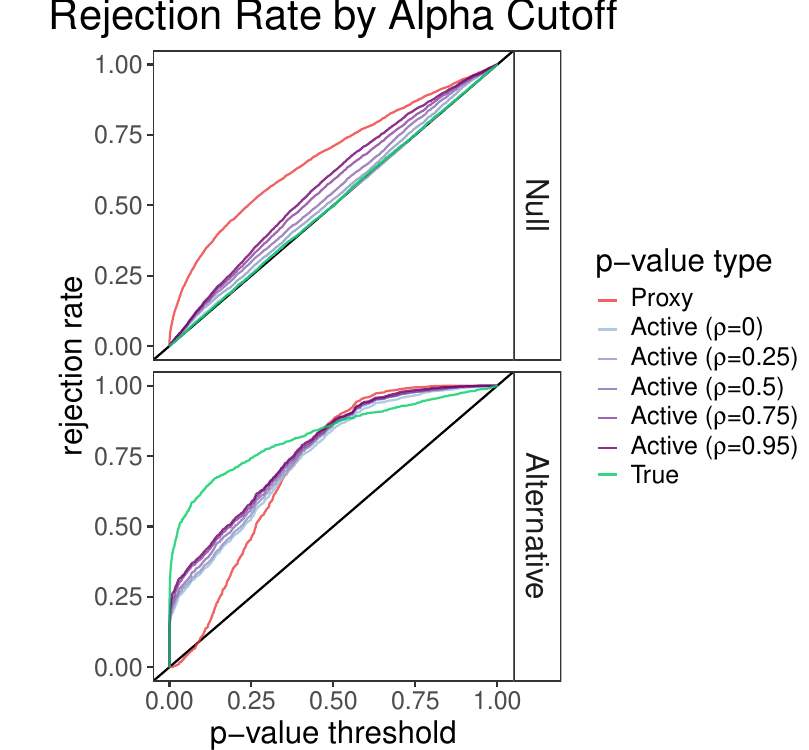}
        \caption{Correlated $P$ and $Q$ ($\rho$ is rank correlation)}
        \label{fig:simBeta_cor}
    \end{subfigure}

    \caption{
Simulation results for independent and rank correlated proxy and true p-values.
    The top row shows results for null hypothesis tests, and the bottom row shows results for alternative hypothesis tests.
    The QQ-plots compare the p-values to a Uniform[0,1].
    }\label{fig:simBeta}
\end{figure}

We consider simulation setups similar to the observed features in the scCRISPR screen experiment.
We consider $2000$ hypothesis tests with $1000$ null hypotheses and $1000$ alternative hypotheses, and
we let the marginal distributions of $Q$ (proxy) and $P$ (true) be
\begin{align}
    Q &\sim \text{Beta}(0.5, 1)& \text{ and }&& P &\sim \text{Uniform}[0, 1] \text{ under the null,}\\
    Q &\sim \text{Beta}(2, 5)& \text{ and }&& P &\sim \text{Beta}(0.2, 1) \text{ under the alternative.}
\end{align}
Here, Uniform[0, 1] refers to the uniform distribution over $[0, 1]$. Under the null, the distribution of $Q$ is stochastically smaller than the uniform, so $Q$ itself is not a valid p-value, while $P$ is exactly uniform and a valid p-value.
Under the alternative, both $Q$ and $P$ have power, but $Q$ is less likely to reject and has less power compared to the true p-value $P$.
Our synthetic experiments also explore the effect of varying dependence levels on validity. Our theoretical results only guarantee type I error control under independence, but the simulation results show robustness under mild dependence levels.

\Cref{fig:simBeta_qqplot} compares the distribution of the proxy, true, active $\tilde{P}^{\textnormal{density}}$ (using the known density of $Q$), and active $\tilde{P}^{\textnormal{density}}$ (using the estimated density of $Q$) p-values to a Uniform[0,1].
\Cref{fig:simBeta_cor} shows the c.d.f.\ plot comparing the resulting p-values to a uniform distribution.
We estimate the density with local polynomial density estimation by \cite{cattaneo_2020} on a set of $200$ known null hypothesis tests.
Under the null, both the active p-values using the known density and the active p-values using an estimated density closely follow a uniform distribution.

Additionally, the power of the active p-values typically lies in between the power of the proxy p-values and the power of the true p-values; however, this is not always the case, and the power depends on the null proxy, alternative proxy, and alternative true distributions. These can be seen by changing the shape parameters of the Beta distribution from which these p-values are drawn.

In the scCRISPR screen experiment, the $P$ and $Q$ p-values are correlated because these p-values are calculated on the same data.
So, we also consider a simulation setup where we induce rank correlation between $Q$ and $P$ using the Iman-Conover transformation \citep{iman_distribution_1982,pouillot_evaluating_2010}.  We control the strength of the rank correlation through a parameter $\rho \in [0, 1]$. We include the full details about the Iman-Conover transformation in \Cref{sec:prds-def}.
We consider $\rho \in\{0, .25, .5, .75, 1\}$ for our simulations where $\rho = 0$ corresponds to independent $P$ and $Q$.

\Cref{fig:simBeta_cor} illustrates this procedure across varying correlation strengths
and shows that the validity of the ``density'' active p-value is affected by rank correlation.
As the correlation increases, the active p-values no longer control for type I error, but even with very high correlation of $.95$, the type I error is not too high. For example, with the same simulation setting as the independent proxy and true p-values but with a rank correlation of $.95$ induced, testing at $\alpha=.05$ leads to a false positive rate of $.072$, and the additional false positive rate decreases as the correlation decreases.

\Cref{fig:2StepHypSim_scale_cor} plots the effect of $\eta \in (0,1]$ on the active p-value procedure. This also leads to uniform active p-values but changes the proportion of true hypothesis tests performed, which allows one to potentially target a desired total computation cost. Setting $\eta = 1$ gives the lowest probability of querying the true p-value that still provides an active p-value, but setting smaller $\eta$ values allows one to query the true p-value more often at the cost of greater computation time. For example, if the true p-values are more powerful and the computation time is not too costly, one may desire to query the true p-value more often.

\begin{figure}[ht]
\centering
   \begin{subfigure}[t]{.35\textwidth}
        \includegraphics[height=3.5cm, left, trim=0 0 2.7in .4in, clip]{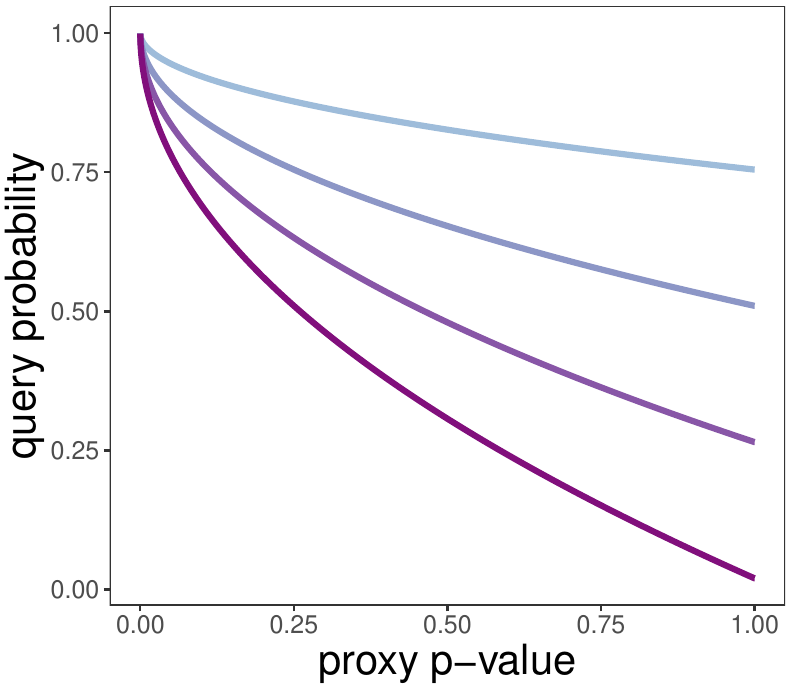}
        \caption{query prob.\ vs.\ proxy p-value}
   \end{subfigure}\qquad
   \begin{subfigure}[t]{.35\textwidth}
        \includegraphics[height=3.5cm, center, trim=0 0 1.9in .4in, clip]{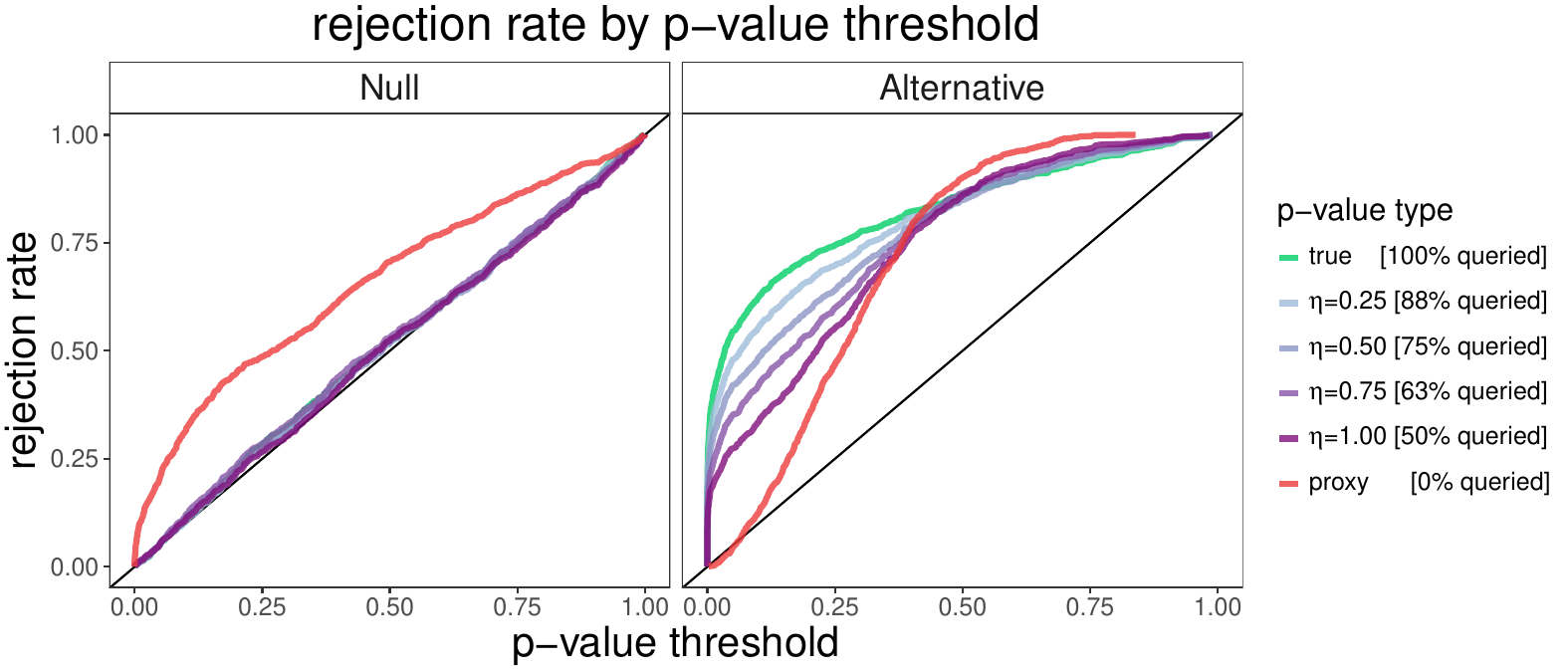}
        \caption{rejection rate vs. p-value threshold}
   \label{fig:2StepHypSim_scale}
   \end{subfigure}\qquad
   \begin{subfigure}[t]{.15\textwidth}
        \includegraphics[height=3.5cm, right, trim=8.5in .75in .125in .75in, clip]{figures/pvalDensityEst/simBeta/setting=03/scaling/scaling_qqplot.pdf}
   \end{subfigure}
   \caption{
       Synthetic experiment result showcasing the scaling factor $\eta$, where $L=\eta l$ is scaled based on an initial lower bound $l$ for $\eta \in [0,1]$. $Q$ and $P$ are independent.
       (a) the probability of querying the true statistic based on the first proxy p-value.
       (b) the rejection rate by p-value threshold.}
   \label{fig:2StepHypSim_scale_cor}
\end{figure}

\section{Estimating the joint distribution of proxy and true p-values\label{sec:joint}}

Instead of assuming independence of $P$ and $Q$, when both the proxy and true p-values are measured on a set of null tests, we may correct for the bias of proxy p-values by using their joint distribution. With the proxy p-value $Q$ and true p-value $P$ ($\text{Uniform}[0,1]$ under the null), let the conditional c.d.f. of the true p-value given the proxy p-value be $\textnormal{F}(p\mid q) = \prob{P \leq p \mid Q = q}$ and the marginal densities be $f_Q(q)$ and $f_P(p)$. Our goal is to construct $\Tilde{P}$ such that $\prob{\Tilde{P} \leq s} = s$ for all  $s \in (0,1)$ using the observed proxies $Q$ and the estimated or known conditional density $\textnormal{F}$. Thus, we can let $\Tilde{P}\mid Q \sim \textnormal{F}(\ \cdot\mid Q)$ and we can plug in the estimated $\hat{\textnormal{F}}$ for $\textnormal{F}$ if do not know $\textnormal{F}$ exactly. In the following, we show that $\Tilde{P}$ obtains the targeted distribution $\sim \text{Uniform}[0,1]$ under the null.

\begin{proposition}
    Consider p-values proxy $Q$ and true $P$. When $P \sim \text{Uniform}[0,1]$, assume that $Q$ is observed, and $F(p\mid q) = \prob{P \leq p \mid  Q = q}$ is known. If we define $\tilde{P}$ s.t.\ $\Tilde{P}\mid Q := F^{-1}(U\mid Q)$, where $U \sim \text{Uniform}[0,1]$ is sampled independently, then $\tilde{P}$ is a valid p-value.
    \label{prop:joint-pvalue-valid}
\end{proposition}

We defer the proof to \Cref{sec:joint-pvalue-proof}.
We notice that none of the expensive $P$ p-values are actually queried for additional hypothesis tests. Instead, there is an initial fixed cost of querying enough $P$ and $Q$ to produce an estimate of the conditional distribution $F$.
Then, the biased proxy p-value is corrected through the conditional distribution.

A disadvantage of this joint distribution procedure is the estimation of the conditional density. Enough of the true p-value must be queried in order to have an adequate estimate of $F(p \mid q)$. In some cases, it may be less computationally expensive to occasionally query the true statistic in the previous active procedures. Additionally, one may want to query true p-values because these true p-values may be deemed desirable. For example, these true p-values may be more interpretable or justified from a causal framework; or these true p-values may have high power.
As a modification, one can decide to randomly query the true p-values based on some fixed chance, $T \sim \text{Bern}(\gamma), \gamma \in [0,1]$, so that the final p-value $\Tilde{P}_2 = (1 - T) \cdot \Tilde{P} + T \cdot P$ also has a $\text{Uniform}[0, 1]$ distribution; however, this does not query for true statistics in an informative way.

\section{Conclusion}
We have introduced a novel framework for active hypothesis testing that allows researchers to efficiently allocate resources by leveraging proxy statistics. We have developed active versions of both e-values and p-values that utilize proxy statistics to reduce resource usage for testing many hypotheses. We also show how these active formulations mesh with FDR controlling multiple testing procedures and consequently provide active versions of the e-BH and BH procedures for multiple testing.
We demonstrate the practical utility of our framework through numerical simulations and an application to testing for causal effects in scCRISPR screen experiments.

While this paper provides error control guarantees and demonstrates the practical significance of our active testing framework, theoretical understanding of how to choose proxy statistics to maximize power remains an open question. We hope this paper provides a starting point for understanding how to perform hypothesis testing and multiple testing under resource constraints.

\bibliography{ref}
\appendix

\section{Extensions to active methods}\label{sec:extensions}

In the following section, we describe a couple of different querying strategies for producing active e-values in the multiple testing setting where we have multiple hypotheses, each with its own proxy e-value. These methods demonstrate the flexibility of the active multiple testing framework and how it can incorporate multilevel approximations of the proxy e-values and human-in-the-loop improvement of remaining proxy statistics when one has queried some true statistics.

\subsection{Multilevel active e-values}
While our basic setup of FDR control with active e-values originally did one round of proxy construction, we can repeatedly calculate our proxies (e.g., iterates of an optimization algorithm) in a fashion that is dependent on the other proxies and queried e-values. For example, we might have estimators that can produce proxy e-values whose accuracy can be improved with more computation, and the true e-value requires an exhaustive search over a grid of inputs. This type of behavior is often seen in Markov Chain Monte Carlo (MCMC) methods, where more computation (i.e., sampling steps), improves the approximation of a test statistic. A particular connection between our method and MCMC methods is to finite sample, unbiased MCMC methods introduced in \citet{rhee_unbiased_estimation_2015,donmcleish_general_method_2011a}. An unbiased estimator is constructed by randomly choosing whether to proceed one more iteration, where the randomness is independent of the MCMC sample itself. Inspired by this, we also formulate a version of a proxy method, where one can iteratively derive a new proxy e-value, and stop the computation according to a stopping rule (i.e., a decision rule that determines which iteration to stop on based on the observed proxy e-values). We call the resulting e-values \emph{multilevel active e-values}, and we formulate our algorithm for computing them in \Cref{alg:multilevel-e-value}.
\begin{algorithm}
    \caption{Multilevel active e-value computation.}\label{alg:multilevel-e-value}
    \SetKw{Break}{Break}
    \KwIn{Sequences of proxy e-values derived from an iterative method $(F_{i}^{(t)})_{t \in \naturals}$ for each $i \in [K]$ and stopping rules $\tau_i$ for each $i \in [K]$. Tuning parameter $\gamma \in (0, 1]$.}
Initialize stopped indices $\Scal_0 \coloneqq \emptyset$.\\
\For{$t \in \{1, 2, \dots,\}$}{Compute $F^{(t)}_i$ for each $i \in [K] \setminus \Qcal_{t - 1}$.\\
Set $\Qcal_t \coloneqq \Qcal_{t - 1}$.\\
\For{$ i \in [K] \setminus \Qcal_{t - 1}$}{
            $\Qcal_t \coloneqq \Qcal_t \cup \{i\} $ \textbf{ if } $\tau_i = t$ based on $(F_i^{(j)})_{i \in [K], j \in [t \wedge \tau_i]}$.
        }
    }
    \For{$i \in [K]$}{Sample $T_i \sim \text{Bern}((1 - \gamma(F_{i}^{(\tau_i)})^{-1})_+)$.\\
        Set $\tilde{E}_i \coloneqq (1 - \gamma)E_i \textbf{ if }T_i = 1 \textbf{ else } F^{(\tau_i)}_i$.
    }
    \KwOut{E-values $(\tilde{E}_1, \dots, \tilde{E}_K)$.}
\end{algorithm}
\begin{proposition}
    $(\tilde{E}_1, \dots, \tilde{E}_K)$ that are produced by \Cref{alg:multilevel-e-value} are valid e-values.
\end{proposition}
The proof follows from the fact that the marginal distribution of $F_i^{(\tau_i)}, E_i,$ and $T_i$ is identical to $F, E,$ and $T$ in \eqref{eq:active-e-value}, so $\tilde{E}_i$ is a valid active e-value by \Cref{prop:active-e-value}.

\subsection{Inter-active e-values}
As we receive more information from running experiments and computing e-values, we can update our remaining proxy e-values. We can follow the procedure in \Cref{alg:interactive-e-value} to decide which experiments to run and calculate our active e-values:
\begin{algorithm}[ht]
    \label{alg:interactive-e-value}
    \caption{Interactive active e-values that can use queried true e-values to inform predictions.}
\KwIn{Initial proxy e-values $(F_1, \dots, F_K)$.}
    Initialize $(F_1^{(0)}, \dots, F_K^{(0)}) \coloneqq (F_1, \dots, F_K)$.\\
    Initialize queried set $\mathcal{Q}_0 \coloneqq \emptyset$.\\
\For{$t \in [K]$}{
        Select $I_t \in [K] \setminus \mathcal{Q}_t$.\\
        Let $T_{I_t} \sim \text{Bern}((1 - \gamma(F^{(t)})^{-1})_+)$.\\
Compute $E_{I_t}$ if $T_{I_t} = 1$.\\
Set $\tilde{E}_{I_t} \coloneqq F^{(t)}_{I_{t}} \textbf{ if }T_{I_t} = 0 \textbf{ else } (1 - \gamma) \cdot E_{I_t} $.\\
Let $\mathcal{Q}_{t + 1} \coloneqq \mathcal{Q}_t \cup \{t + 1\}$.\\
Update $(F^{(t + 1)}_i)_{i \not \in \mathcal{Q}_{t + 1}}$ based on $(T_i)_{i \in \mathcal{Q}_{t + 1}}$ and $(E_i)_{i \in \mathcal{Q}_{t + 1}, T_i = 1}$.\\
Let $F_i^{(t +1)} \coloneqq F_i^{(t)}$ for $i \in \mathcal{Q}_{t + 1}$.\\
    }
\end{algorithm}
\begin{proposition}
    $(\tilde{E}_1, \dots, \tilde{E}_K)$ as produced through \Cref{alg:interactive-e-value} are  e-values.
\end{proposition}
Similarly, each active e-value defined in \Cref{alg:interactive-e-value}, conforms to the definition in \eqref{eq:active-e-value}, so each inter-active e-value is a bona fide e-value. Thus, inter-active e-values allow one to update proxy e-values as the scientist gathers more data through querying the true e-values for a subset of the considered hypotheses.
Naturally, one can also construct p-value analogs of the aforementioned procedures. However, the dependence structures of the resulting p-values are difficult to decipher, and hence ensuring valid FDR control at level $\alpha$ might require one to inflate the test level of BH. On the other hand, e-BH permits any dependence structure, and thus can be directly applied to these processed e-values with FDR control.

\section{Additional details for scCRISPR screen experiments}
\label{sec:appendix:scCRISPR}

\subsection{Proximal inference assumptions}\label{sec:appendix:scCRISPR:proximalassumptions}
We describe the proximal inference setup and assumptions in more detail. Proximal inference focusing on the outcome regression model makes four main assumptions \cite{miao_confounding_2020}:
\begin{assumption}[Proximal causal inference assumptions]
    We make the following additional assumptions about the $Y, A, W$, and $Z$.
\begin{enumerate}
    \item (negative control outcome) $W \indep A \mid C$ and $W \not\indep C$ .
    \item (negative control exposure) $Z \indep Y \mid C, A$ and $Z \indep W \mid C$.
    \item (outcome confounding bridge) There exists a function $h$ such that $\expect[Y \mid C, A=a] = \expect[h(W, a)  \mid  C, A=a]$ for $a \in \{0, 1\}$.
    \item Completeness of $\sprob(W \mid Z, A)$: For all $a$, $W \not\indep Z  \mid  A=a$; and for any square integrable function $g$, if $\expect\left[g(W)  \mid  Z=z, A=a\right] = 0$ for almost all $z$, then $g(W) = 0$ almost surely.
\end{enumerate}
\end{assumption}
With $W$ and $h$, the mean potential outcome is identified by $\expect(Y^a) = \expect[h(W, a)]$. However, this $h$ cannot be identified yet because this equation characterizing $h$ involves unmeasured variables $C$. If the completeness condition holds and there exists a negative control exposure $Z$, then $h$ can be estimated with $\expect\left[Y \mid Z, A\right] =  \expect\left[h(W, A\right)  \mid  Z, A]$. \Cref{fig:proxDAG} depicts a causal directed acyclic graph (DAG) with random variables as nodes $(A, Y, C, Z, W)$ and arrows indicating causal effects between variables. This causal DAG satisfies the proximal inference setting but is not the only DAG that does.

\subsection{Methodological details of 2SLS}\label{sec:appendix:scCRISPR:2SLSdetails}
We describe the proximal inference method and procedure by \cite{liu_regression-based_2024} in more detail. The following is a summarization of the linear case in Appendix A.9 of \cite{liu_regression-based_2024}. The authors also consider situations when the negative control outcome and outcome variables may have different forms, such as binary and count forms. In our application, we consider continuous negative control outcomes and outcome, which corresponds to identity link functions $g_1(x) = x$ and $g_2(x) = x$, but the following results hold for other link functions such as log and logit links. Additionally, they consider including measured confounders $X \in \reals^{\# \text{measured confounders}}$. Our application does not use measured confounders, so to simplify notation, we exclude these $X$. Lastly, they describe an example when using $1$ negative control, but the summary below describes an example with multiple negative controls.

Let $d$ be the number of negative control exposures as well as the number of negative control outcomes, $Z \in \mathbb{R}^d$ and $W \in \mathbb{R}^d$. Suppose the true parameters are $\alpha_{j,0}^*, \alpha_{j,a}^* \in \reals$ and $\alpha_{j,z}^*\in\reals^d$ for $j=1,...,d$ and $\beta_0^*, \beta_a^* \in \reals$ and $\beta_u^* \in \reals^d$. Let $$\alpha^* = [\alpha_{1,0}^*, \alpha_{1,a}^*, \alpha_{1,z}^{*T},..., \alpha_{d,0}^*, \alpha_{d,a}^*, \alpha_{d,z}^{*T}]^T \in \reals^{d(2+d)} = \reals^{2d+d^2}$$ and $\beta^* = [\beta_0^*, \beta_a^*, ..., \beta_u^{*T}]^T \in \reals^{2+d}$. Then, the two stages of least squares regression are modeled as:

\begin{align}
    \text{First Stage:}
& & \expect[W_j \mid A,Z]
        &= \alpha_{j,0}^* + \alpha_{j,a}^* A + \alpha_{j,z}^{*T} Z \quad\text{ for } j=1,...,d \\
    \text{Second Stage:}
& & \expect[Y \mid A,Z]
        &= \beta_0^* + \beta_a^* A + \beta_u^{*T} S(\alpha^*)\\
    \text{For the Proximal Control Variable:}
    & & S_j = S_j(\alpha_j^*)
        &= \alpha_{j,0}^* + \alpha_{j,a}^* A + \alpha_{j,z}^{*T} Z \\
    & & S = S(\alpha^*)
        &= [S_1, ..., S_d]^T \in \mathbb{R}^{d} \\
\end{align}

Both least square stages are estimated using maximum likelihood estimation, so the estimates may be analyzed asymptotically as m-estimators.
Denote the observed random variables variables $\mathcal{O}:=(A,Y,Z,W)$.
The first stage's estimating equation is $\expect[\Psi_{1,j}(\mathcal{O};\alpha^*, \beta^*)| A,Z] = 0$  for $j = 1, ..., d$, and the second stage's estimating equation is $\expect[\Psi_2(\mathcal{O};\alpha^*, \beta^*) | A,Z] = 0$,
where $\Psi_{1,j}$ and $\Psi_2$ are:
\begin{align}
    \text{First Stage:}
        & & \Psi_{1,j} = \Psi_{1,j}(\mathcal{O};\alpha^*, \beta^*)
        &= \begin{bmatrix}
            1 \\ A \\ Z
            \end{bmatrix}
            \left\{
                W_j -
(\alpha_{j,0}^* + \alpha_{j,a}^* A + \alpha_{j,z}^{*T} Z)
            \right\}
            \in \mathbb{R}^{2 + d}\\
    \text{Second Stage:}
        & &\Psi_2 = \Psi_2(\mathcal{O};\alpha^*, \beta^*)
        &= \begin{bmatrix}
            1 \\ A \\ S
            \end{bmatrix}
            \left\{
                 Y -
(\beta_0^* + \beta_a^* A + \beta_u^* S)
            \right\}
            \in \mathbb{R}^{2 + d}
\end{align}
Define $\Psi = [\Psi_{1,1}^{T}, ..., \Psi_{1,d}^{T}, \Psi_2^T]^T \in \mathbb{R}^{(d+1)(2 + d)} = \mathbb{R}^{2 + 3d + d^2}$. Together, the estimating equations form $\expect[\Psi|A,Z] = 0$, and the estimates $(\hat\alpha, \hat\beta)$ minimize the empirical average $\mathbb{P}_n\{\Psi(\mathcal{O}; \alpha, \beta)\} = 0$, where the notation $\mathbb{P}_n$ means $\mathbb{P}_n\{X\} = (1/n)\sum_{i=1}^n X_i$.
Under certain smoothness assumptions, explained in detail in \cite{liu_regression-based_2024}, Appendix A.9, the estimators are asymptotically normal and the variance may be estimated by sandwich variance estimators.
\begin{align}
    \sqrt{N}\left\{
    \begin{pmatrix}
        \hat\alpha \\
        \hat\beta
    \end{pmatrix}
    -
    \begin{pmatrix}
        \alpha^* \\
        \beta^*
    \end{pmatrix}
    \right\}
    \rightarrow^d
    N(0, \mathbf{V}(\alpha^*, \beta^*))
\end{align}
where
\begin{align}
\mathbf{V}(\alpha^*, \beta^*) &= \mathbf{A}(\alpha^*, \beta^*)^{-1}
                                     \mathbf{B}(\alpha^*, \beta^*)
                                     \{\mathbf{A}(\alpha^*, \beta^*)^{-1}\}^T \\
\mathbf{A}(\alpha^*, \beta^*) &= \expect\left[\left.
                                            \frac{\partial}{\partial(\alpha, \beta)^T}\Psi(\mathcal{O}; \alpha, \beta)
                                            \right|_{\alpha=\alpha^*, \beta=\beta^*}
                                          \right]
    \in \mathbb{R}^{(2 + 3d + d^2) \times (2 + 3d + d^2)} \\
    \mathbf{B}(\alpha^*, \beta^*) &= \expect[\Psi(\mathcal{O}; \alpha^*, \beta^*)\Psi(\mathcal{O}; \alpha^*, \beta^*)^T]
    \in \mathbb{R}^{(2 + 3d + d^2) \times (2 + 3d + d^2)}
\end{align}

Finally, the standard error, $\hat\sigma^\TSLS$, of the ATE estimate, $\hat\beta_a$, can be calculated by estimating $\mathbf{V}(\alpha^*, \beta^*)$, indexing into the covariance for $\beta_a$, and taking the square root. $\mathbf{V}(\alpha^*, \beta^*)$ may be estimated by taking empirical averages and plugging in the estimates for $(\alpha^*, \beta^*)$ in the components $\mathbf{A}(\alpha^*, \beta^*)$ and $\mathbf{B}(\alpha^*, \beta^*)$.
\begin{align}
    \hat\sigma^\TSLS &= \sqrt{\mathbf{V}_n(\hat\alpha, \hat\beta)_{A,A}} \\
    \text{ where }\qquad
    \mathbf{V}_n(\hat\alpha, \hat\beta) &= \mathbf{A}_n(\hat\alpha, \hat\beta)^{-1}
                                     \mathbf{B}_n(\hat\alpha, \hat\beta)
                                     \{\mathbf{A}_n(\hat\alpha, \hat\beta)^{-1}\}^T \\
    \mathbf{A}_n(\hat\alpha, \hat\beta) &= \mathbb{P}_n\left[\left.
                                            \frac{\partial}{\partial(\alpha, \beta)^T}\Psi(\mathcal{O}; \alpha, \beta)
                                            \right|_{\alpha=\hat\alpha, \beta=\hat\beta}
                                          \right] \\
    \mathbf{B}_n(\hat\alpha, \hat\beta) &= \mathbb{P}_n[\Psi(\mathcal{O}; \hat\alpha, \hat\beta)\Psi(\mathcal{O};  \hat\alpha, \hat\beta)^T]
\end{align}

Although the two stages of least squares are relatively fast, because this involves just $d+1$ linear regressions, the calculation of the variance takes much longer and is the source of the 2SLS approach's high cost.

\paragraph{Runtime Complexity in terms of $n$ and $d$.}
We compare the runtime complexities of the fast but biased test (linear regression of $Y$ on $A$) and the slower but true test (proximal 2SLS).
In general, linear regression with $p$ covariates has a runtime complexity of $O(np^2 + p^3)$.

In the fast but biased test, $p$ is fixed and small at $p=2$, a coefficient for the constant and a coefficient for the treatment assignment, so overall, the fast linear regression has a runtime of $O(n2^2 + 2^3)=O(n)$.
The runtime complexity for estimating the parameter is $O(np^2 + p^3)$ from computing
$\hat\beta^\OLS \coloneqq (\mathbf{M}^T \mathbf{M})^{-1} \mathbf{M}^T \vv{\mathbf{Y}}\text{ where } \mathbf{M} = [\mathbf{1}, \vv{\mathbf{A}}]$.
And the runtime complexity for calculating the standard error and p-value is $O(np^2 + p^3)$ from
$\hat{\sigma}^\OLS \coloneqq \sqrt{\left((\mathbf{M}^T \mathbf{M})^{-1}  \hat\sigma^2\right)_{A,A}}$ where $\hat\sigma = \sqrt{\sum_i^n (Y - \mathbf{M}\beta)^2/(n-2)}$.

In comparison, the time complexity of the proximal 2SLS approach is $O(nd^4 + d^6)$.
Estimating the ATE takes $O((d+1)(n(d+2)^2 + (d+2)^3)) = O(nd^3 + d^4)$, because there are $d+1$ least squares estimations each with $d+2$ covariates.
Calculating the standard error and p-value takes $O(nd^4 + d^6)$. More details are described below.

\paragraph{More details about calculating $\mathbf{V}_n(\hat\alpha, \hat\beta)$.}
To better understand the runtime complexity of proximal 2SLS, let us focus on $\mathbf{A}(\alpha^*, \beta^*)$.
Below, we show that constructing each element of the partial derivative matrix takes constant time $O(1)$; then, we describe the total runtime for estimating $\mathbf{V}(\alpha^*, \beta^*)$ with $\mathbf{V}_n(\hat\alpha, \hat\beta)$.

Let us combine the parameters together $\theta = [\alpha^T, \beta^T]^T \in \reals^{2 + 3d + d^2}$. Each element $k,l \in [2+3d+d^2]$ in the partial derivative matrix is:
\begin{align}
    \left(\frac{\partial}{\partial\theta}\Psi(\mathcal{O}; \theta)\right)_{k,l}
    &= \frac{\partial}{\partial\theta_l}\Psi(\mathcal{O}; \theta)_k \\
    &= \mathbb{I}\left\{\left\lfloor \frac{k}{d}  \right\rfloor = \left\lfloor \frac{l}{d}  \right\rfloor\right\}
       \cdot G(l) \cdot H(k)
\end{align}

where
\begin{itemize}
    \item $\mathbb{I}\left\{\left\lfloor \frac{k}{d}  \right\rfloor = \left\lfloor \frac{l}{d}  \right\rfloor\right\}$ indicates whether $\Psi(\mathcal{O}; \theta)_k$ and $\theta_l$ correspond to the same estimating equation. If they belong to separate estimating equations, then this partial derivative is $0$.
    \item $G: [2+3d+d^2] \rightarrow \{1, A, Z_1, ..., Z_d, S_1, ..., S_d\}$ returns the corresponding term for $\theta_j$
    \begin{align}
        G(l) &=
        \left\{
        \begin{matrix*}[l]
            1   &\text{ if } \theta_l \text{ is } \alpha_{1, 0}, ..., \alpha_{d, 0}, \text{ or } \beta_0 \\
            A   &\text{ if } \theta_l \text{ is } \alpha_{1, a}, ..., \alpha_{d, a}, \text{ or } \beta_a \\
            Z_j &\text{ if } \theta_l \text{ is } \alpha_{1, z, j}, ..., \alpha_{d, z, j} \text{ for } j=1,...,d \\
            S_j &\text{ if } \theta_l \text{ is } \beta_{z,j} \text{ for } j=1,...,d
        \end{matrix*}
        \right.
    \end{align}
\item $H: [2+3d+d^2] \rightarrow \{1, A, Z_1, ..., Z_d, S_1, ..., S_d\}$ returns the corresponding term for $\Psi_k$
    \begin{align}
        H(k) &=
        \left\{
        \begin{matrix*}[l]
            1   &\text{ if } \Psi_k \text{ is } (\Psi_{1, 1})_1, ..., (\Psi_{1, d})_1, \text{ or } (\Psi_2)_1 \\
            A   &\text{ if } \Psi_k \text{ is } (\Psi_{1, 1})_2, ..., (\Psi_{1, d})_2, \text{ or } (\Psi_2)_2 \\
            Z_j &\text{ if } \Psi_k \text{ is } (\Psi_{1, 1})_{j+2}, ..., (\Psi_{1, d})_{j+2} \text{ for } j=1,...,d \\
            S_j &\text{ if } \Psi_k \text{ is } (\Psi_2)_{j+2} \text{ for } j=1,...,d
        \end{matrix*}
        \right.
    \end{align}

\end{itemize}

The partial derivative $\frac{\partial}{\partial(\alpha, \beta)^T}\Psi(\mathcal{O}; \alpha, \beta)$ of estimating equation $\Psi$ has a block diagonal form:
\begin{align}
    \begin{bmatrix}
        \frac{\partial}{\partial\alpha_1}\Psi_{1,1}(\mathcal{O};\alpha, \beta) & 0 & ... & 0 & 0 \\
        0 & \frac{\partial}{\partial\alpha_2}\Psi_{1,2}(\mathcal{O};\alpha, \beta) & ... & 0 & 0 \\
        \vdots & \vdots & \ddots & \vdots & \vdots\\
        0 & 0 & ... & \frac{\partial}{\partial\alpha_d}\Psi_{1,d}(\mathcal{O};\alpha, \beta) & 0 \\
        0 & 0 & ... & 0 & \frac{\partial}{\partial\beta} \Psi_2(\mathcal{O};\alpha, \beta)
    \end{bmatrix}
\end{align}

This means that the time complexity for estimating $\mathbf{A}(\alpha^*, \beta^*)$ with $\mathbf{A}_n(\hat\alpha, \hat\beta)$, which is a $(2 + 3d + d^2) \times (2 + 3d + d^2)$ matrix, takes $O(n(2 + 3d + d^2)^2)=O(nd^4)$.
There are $n$ samples, and it takes $O((2 + 3d + d^2)^2)$ to form each matrix.
The same idea applies to estimating $\mathbf{B}(\alpha^*, \beta^*)$ with $\mathbf{B}_n(\hat\alpha, \hat\beta) \in \reals^{(2 + 3d + d^2) \times (2 + 3d + d^2)}$, which also takes $O(n(2 + 3d + d^2)^2)=O(nd^4)$ to construct.
Additionally, the $\mathbf{A}_n(\hat\alpha, \hat\beta)$ must be inverted, which may take $O((2 + 3d + d^2)^3)=O(d^6)$.
Lastly, multiplying together to estimate $\mathbf{V}_n(\hat\alpha, \hat\beta)$ takes $O(2(2 + 3d + d^2)^3)=O(d^6)$.
So overall, the runtime complexity for calculating $\mathbf{V}_n(\hat\alpha, \hat\beta)$ takes $O(nd^4 + d^6)$.

As an aside, because $\mathbf{A}_n(\hat\alpha, \hat\beta)$ has a block diagonal form, it is possible to reduce the computation time by inverting each of the $d+1$ blocks of size $(d+2) \times (d+2)$ separately. This leads to a runtime of $O((d+1)(d+2)^3)=O(d^4)$ compared to $O(d^6)$. However, this does not reduce the total time complexity, because the $\mathbf{A}_n(\hat\alpha, \hat\beta)$ and $\mathbf{B}_n(\hat\alpha, \hat\beta)$ matrices must still be multiplied together.

\section{Additional experimental results}\label{sec:appendix:exp-results}
We include additional experimental results below that supplement those in \Cref{sec:scCRISPR}. 
\subsection{Joint density active p-value simulation study}
\label{sec:joint-pvalue-sim}

We also investigate the joint density active p-value (\ref{prop:joint-pvalue-valid}) through simulation settings specified in \ref{sec:scCRISPR:synthetic}. 
Figure \ref{fig:jointdens} shows the results of 1000 null simulations and 1000 alternative simulations across a variety of correlation strengths. The conditional density of the true given the proxy is estimated on a separate 1000 draws from the null simulation and estimated with a local polynomial density estimator implemented in the package  \texttt{hdrcde} \citet{hyndman_highest_density_2021}.

In a variety of settings, we see that the active p-values follow a uniform distribution under the null; however, the power of the joint density active p-values is highly reliant on the relative shapes of the proxy and true p-values under the null and alternative settings. Additionally, the power relies on the correlation strength between the proxy and true p-value. Higher correlation tends to lead to higher power, which is expected because the joint density is trying to take advantage of the joint information between the true and proxy p-values, but there is none when they are uncorrelated or independent. 

\begin{figure}[h!t]
   \centering
   \begin{subfigure}[t]{.75\textwidth}
        \includegraphics[height=3.9cm, center, clip, trim={0 0 4.1cm .8cm}]{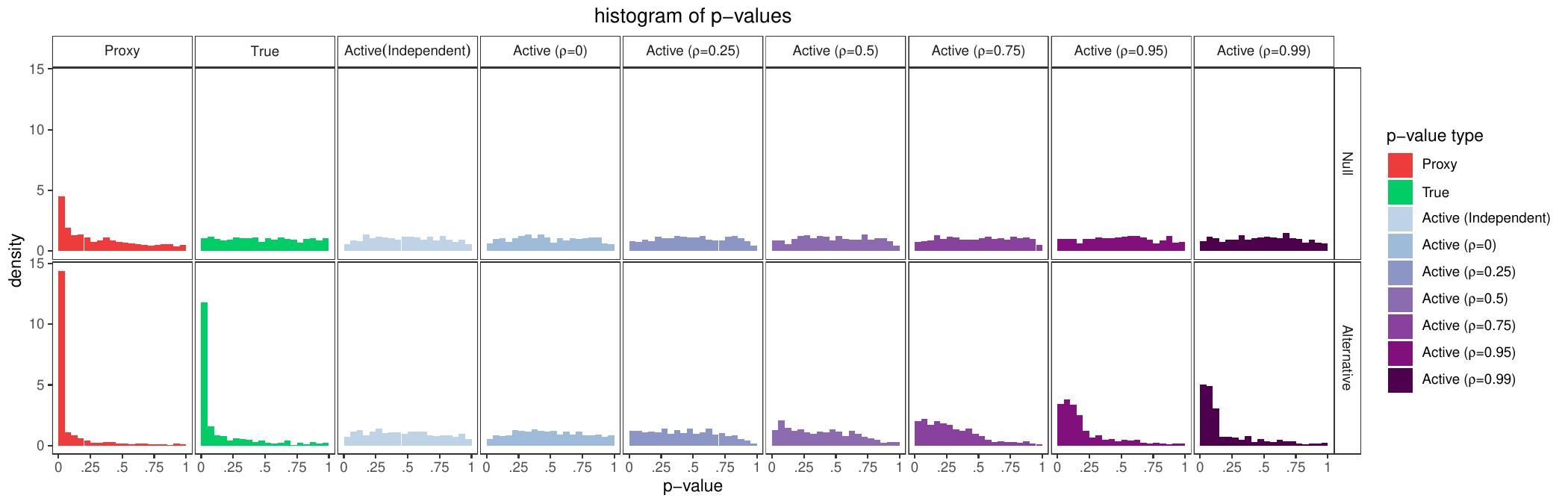}
        \caption*{\qquad\qquad Setting (A)}
   \end{subfigure}\hfill
   \begin{subfigure}[t]{.2\textwidth}
        \includegraphics[height=3.7cm, center, clip, trim={0 0 0 .8cm}]{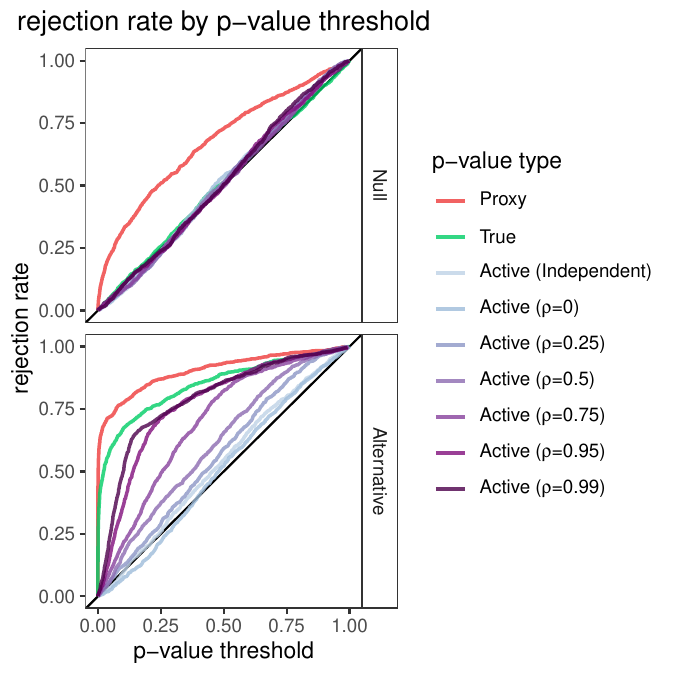}
\end{subfigure}

   \begin{subfigure}[t]{.75\textwidth}
        \includegraphics[height=3.9cm, center, clip, trim={0 0 4.1cm .8cm}]{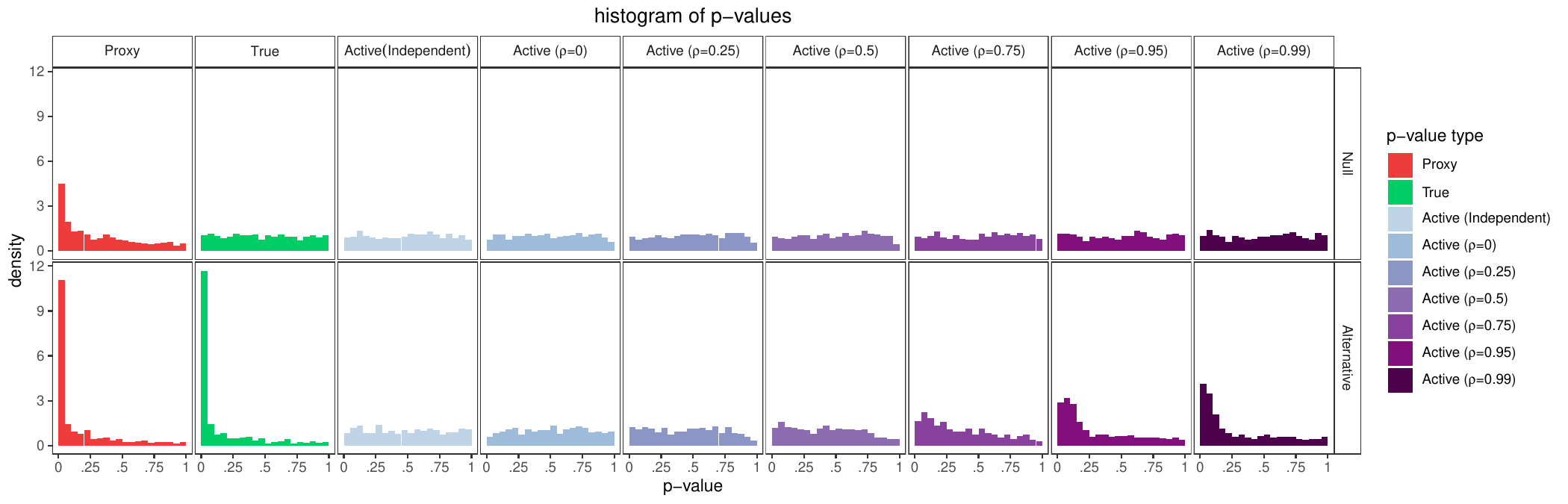}
        \caption*{\qquad\qquad Setting (B)}
   \end{subfigure}\hfill
   \begin{subfigure}[t]{.2\textwidth}
        \includegraphics[height=3.7cm, center, clip, trim={0 0 0 .8cm}]{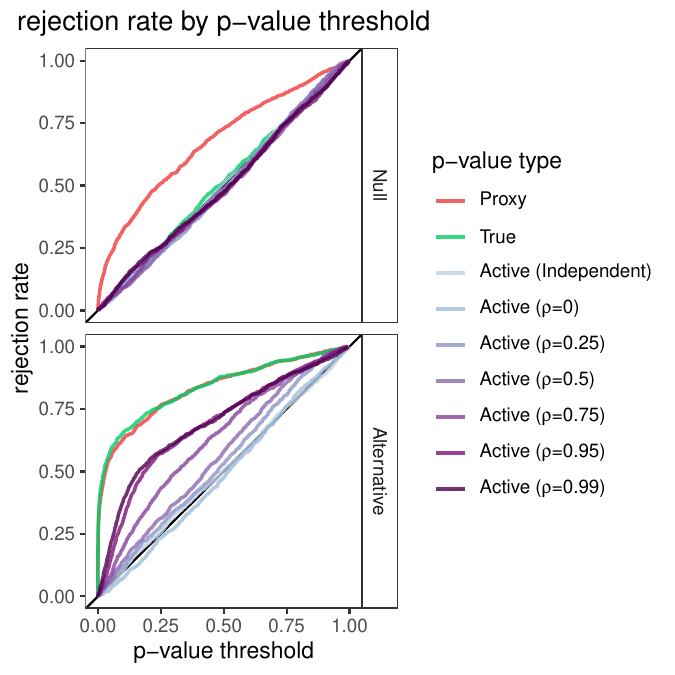}
\end{subfigure}

      \begin{subfigure}[t]{.75\textwidth}
        \includegraphics[height=3.9cm, center, clip, trim={0 0 4.1cm .8cm}]{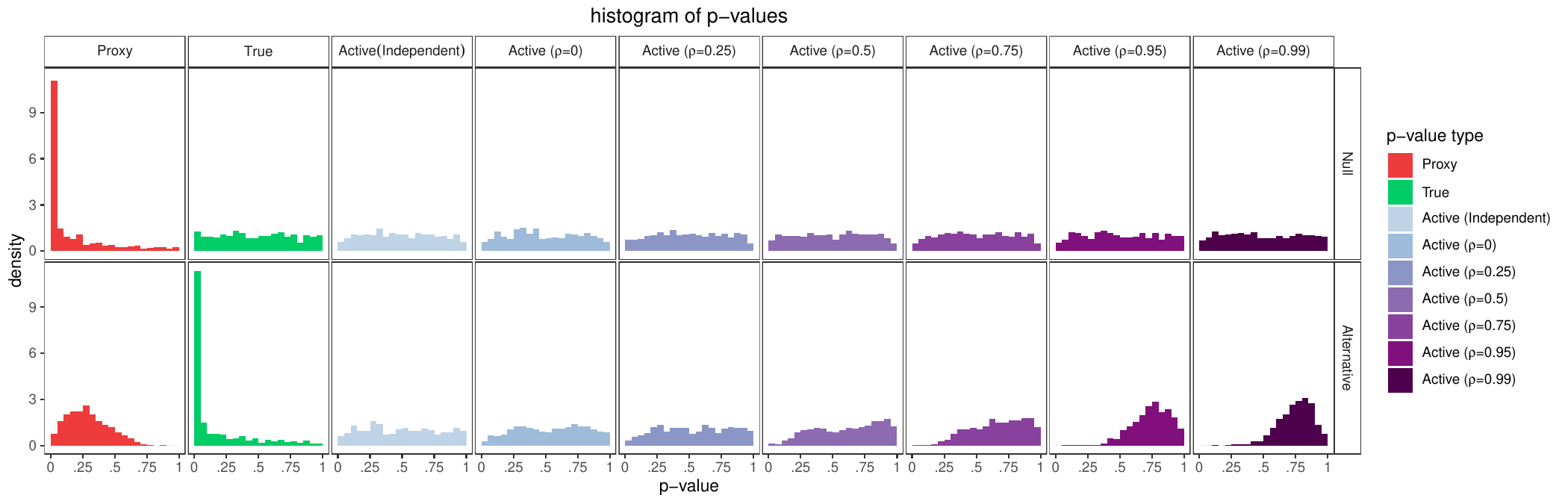}
        \caption*{\qquad\qquad Setting (C)}
   \end{subfigure}\hfill
   \begin{subfigure}[t]{.2\textwidth}
        \includegraphics[height=3.7cm, center, clip, trim={0 0 0 .8cm}]{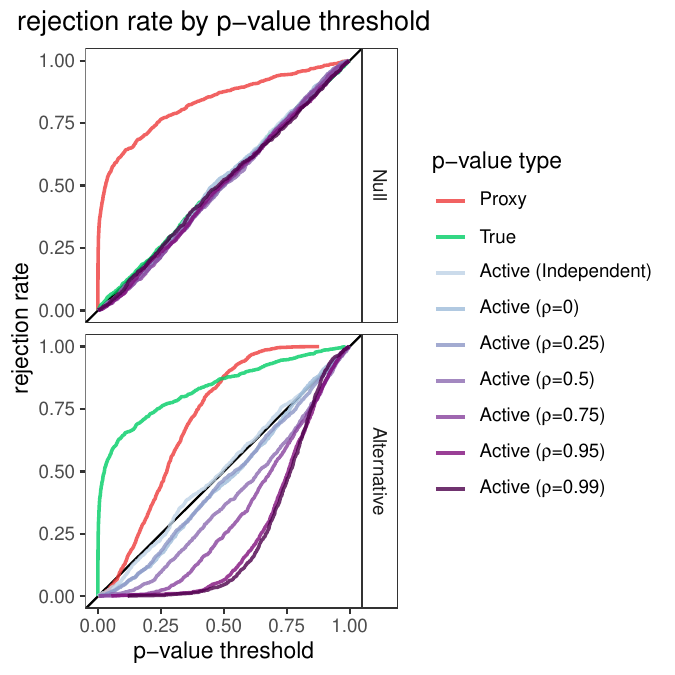}
\end{subfigure}
   
   \caption{
      Simulation results for the joint density active p-value (\ref{prop:joint-pvalue-valid}) under simulation settings described in \ref{sec:scCRISPR:synthetic}. The left plots show the distribution of different p-values under the null and alternative hypothesis setting. The p-values displayed are the proxy, true, and active joint density p-values under various correlation strengths. The right plots show a QQ-plot comparing the different p-values to the Uniform[0,1] distribution.}
   \label{fig:jointdens}
\end{figure}

\subsection{Active p-values in the scCRISPR experiment}
\label{sec:arbdep-papalexi}
In \Cref{fig:arbdep-papalexi} we plot various views of the result of applying the active p-value that is valid under arbitrary dependence \eqref{eq:arbdep-active-pvalue} to the \citet{Papalexi2021} scCRISPR experiment described in \Cref{sec:scCRISPR}. We vary the $\gamma$ parameter that controls the probability of querying the true p-value. As $\gamma$ increases, the active p-values more closely resemble the true p-values, as seen in the QQ-plot in \Cref{fig:arbdep-papalexi}(b). The histogram of the active p-values in \Cref{fig:arbdep-papalexi}(a) also shows that as $\gamma$ increases, the active p-values become more concentrated near $0$, indicating higher power. At the same time, the histogram of query probabilities in \Cref{fig:arbdep-papalexi}(c) shows that as $\gamma$ increases, the probability of querying the true p-value also increases. Thus, we can see in \Cref{fig:arbdep-papalexi}(d) that as $\gamma$ increases, the runtime of the active method approaches that of the true method, since more true p-values are queried.

\begin{figure}[h!t]
   \centering
   \begin{subfigure}[t]{.47\textwidth}
        \includegraphics[height=4.6cm, center, clip, trim={0 0 0 .8cm}]{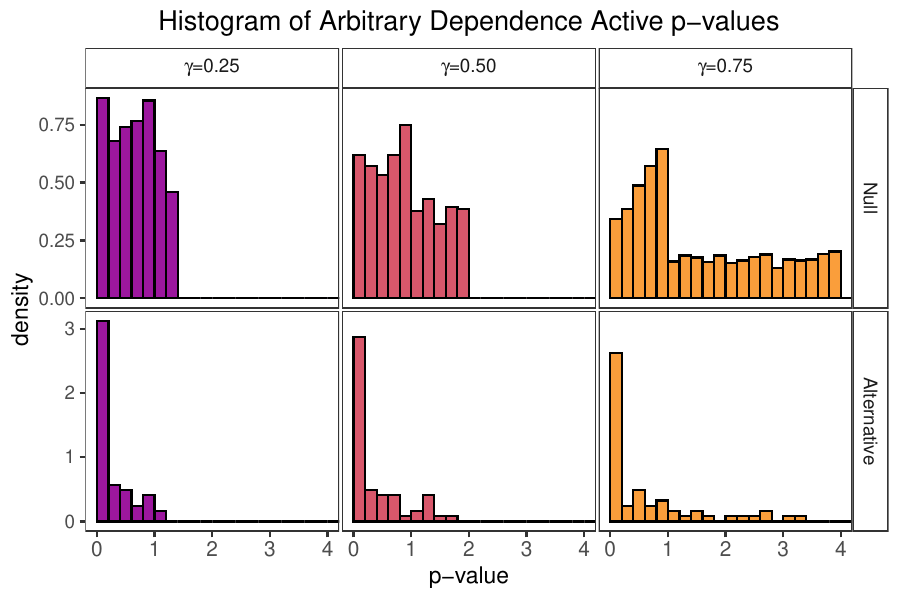}
        \caption{Histogram of arbitrary dependent p-values.}
   \end{subfigure}\qquad
   \begin{subfigure}[t]{.47\textwidth}
        \includegraphics[height=4.6cm, center, clip, trim={0 0 0 .8cm}]{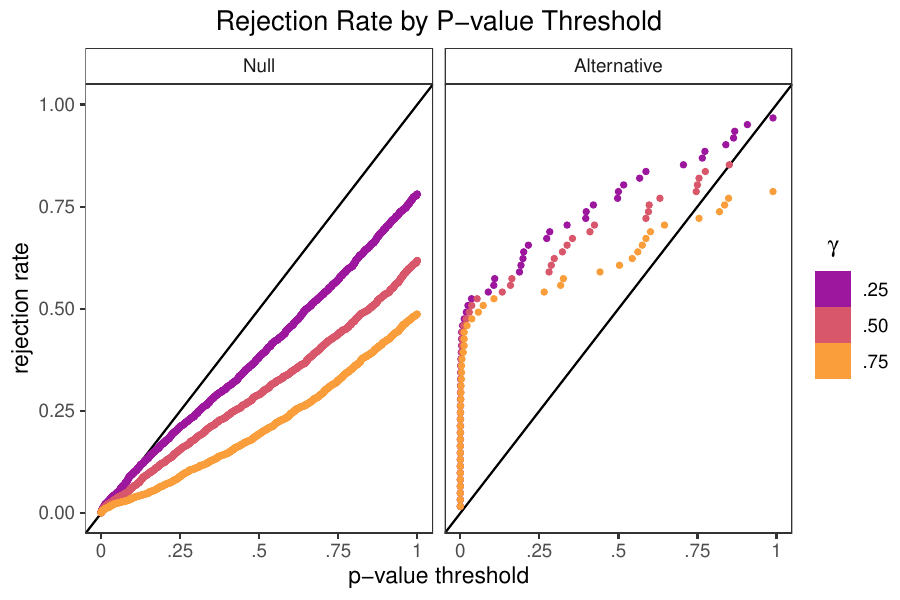}
        \caption{Rejection rate by p-value threshold (cutoff at 1).}
   \end{subfigure}

   \begin{subfigure}[t]{.47\textwidth}
        \includegraphics[height=4.6cm, center, clip, trim={0 0 0 .8cm}]{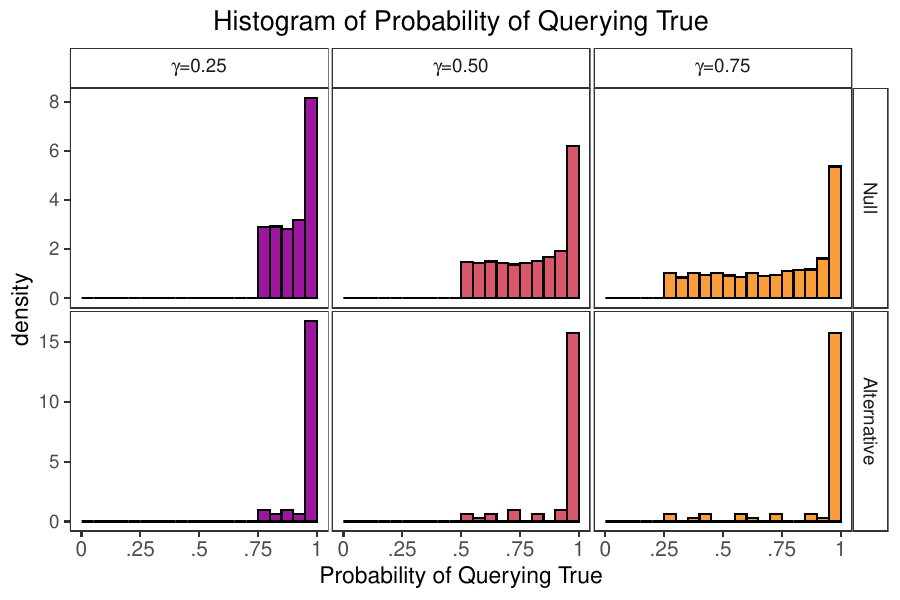}
        \caption{Histogram of the probabilities of querying the true p-value.} \end{subfigure}\qquad
   \begin{subfigure}[t]{.47\textwidth}
        \includegraphics[height=4.6cm, center, clip, trim={0 0 0 .8cm}]{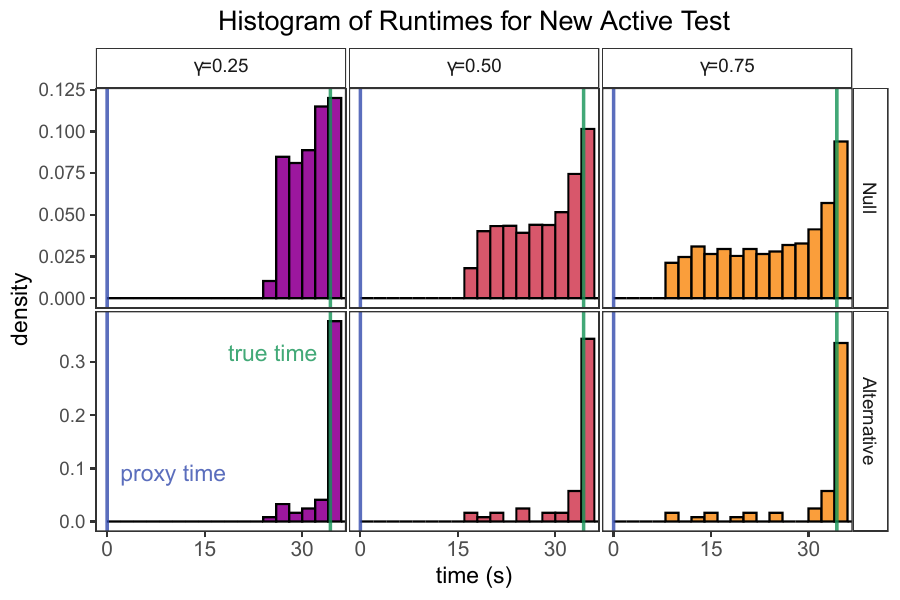}
        \caption{Histogram of runtimes.}
   \end{subfigure}
   
   \caption{
      Arbitrary dependence active p-values applied to the \citet{Papalexi2021} scCRISPR experiment across varying $\gamma$ parameters.}
   \label{fig:arbdep-papalexi}
\end{figure}
\section{Definitions of different types of dependence\label{sec:prds-def}}
We recall the definitions of two different notions of dependence. The first is the PRDS condition for FDR control of the BH procedure that was introduced by \citet{benjamini_control_false_2001}.
\begin{definition}
    A set $D\subseteq \mathbb{R}^K$ is called \emph{increasing} if for any $\boldsymbol{x} \in D$, any $\boldsymbol{y} \geq \boldsymbol{x}$ (component-wise) also belongs to $D$.
\end{definition}
\begin{definition}\label{def:PRDS}
    Furthermore, a vector of p-values $\boldsymbol{P}=(P_1,\ldots,P_K)$ exhibits \emph{positive regression dependence on a subset (PRDS)} if, for any increasing set $D$ and any $k \in [K]$, the probability $\mathbb{P}(\boldsymbol{P}\in D|P_k\leq x)$ does not decrease as $x$ increases from $0$ to $1$.
\end{definition}
A slightly weaker notion of positive dependence is \emph{positive regression dependence on the nulls} (PRDN) \citep{su_fdr_linking_2018}, which only asks for the null p-values to satisfy the PRDS condition.

\begin{definition}
    A sequence of p-values $\boldsymbol{P}=(P_1,\ldots,P_K)$ is \emph{weakly negatively dependent on nulls (WNDN)} if, for any subset $A \subseteq \Ncal$ and any $s \in [0, 1]$, $\prob{\bigcap_{i \in A} P_i \leq s} \leq \prod_{i \in A} \prob{P_i \leq s}$, i.e., the joint probability does not exceed the product of the individual probabilities of landing below $s$.
\end{definition}

There are many types of negative dependence, but WNDN is a condition that is satisfied by many different types of negative dependence assumptions --- an overview of the different types can be seen in \citet{chi_multiple_testing_2024}.

\paragraph{Iman-Conover rank correlation}We give a quick overview of the Iman-Conover \citep{iman_distribution_1982} transformation and its purpose in our simulations. We use the implementation of this method created by \cite{pouillot_evaluating_2010} in the package \texttt{mc2d}.

Given $n$ i.i.d.\ samples from each of $K$ different distributions, the goal of the Iman-Conover transformation is to induce a target rank correlation matrix among the samples across distributions while preserving the original $K$ different marginal distributions.
The user chooses a target rank correlation matrix $C\in \reals^{K\times K}$ which denotes the resulting correlation between each pair of random variables. Then, the new matrix of samples $\tilde{R}$ is produced, where each of its columns is a permutation of the columns of $R$, in accordance with \Cref{alg:iman-conover}.
\begin{algorithm}[h]\label{alg:iman-conover}
\caption{Iman-Conover transformation for inducing target rank correlation.}
\KwIn{Matrix of scores $R \in \reals^{n \times K}$, target correlation matrix $C \in [0, 1]^{K \times K}$, a vector $\mathbf{a} \coloneqq (a_1, \dots, a_n) \in \reals^n$ of scores.}
\KwOut{Matrix $\tilde{R} \in \reals^{N \times K}$ with rank correlation matrix close to $C$}

Draw $K$ independent permutations of $[n]$ denoted $\pi_1, \dots, \pi_K$.\\
Construct a matrix $R \in \reals^{n \times K}$ and set the entry in the $i$th row and $j$th column as $R_{i, j} \coloneqq a_{\pi_j(i)}$ (ind. permuted col.)\\

Compute lower triangular matrix $P$ via Cholesky decomposition (i.e., $P P^T = C$).\\

Compute $R^* \coloneqq R P^T$.\\

Rearrange the columns of the input matrix to match the ordering of the corresponding column of $R^*$ to obtain $\tilde{R}$.\\

\end{algorithm}

\begin{figure}[ht]
   \centering
   \begin{subfigure}{.9\textwidth}
   \centering
   \includegraphics[height=4.2cm, center]{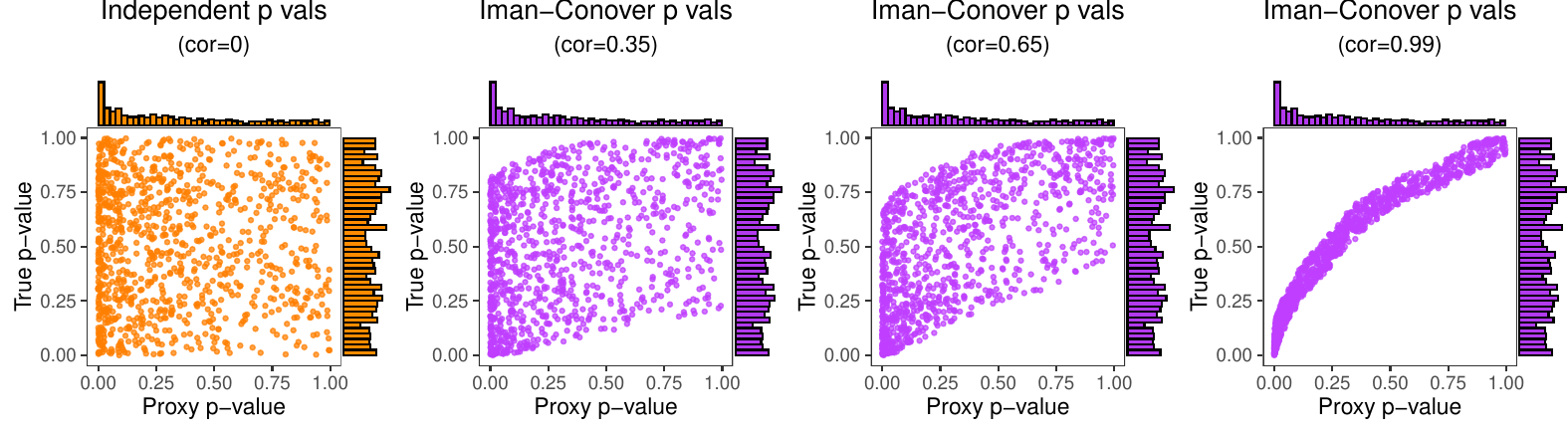}
   \end{subfigure}

   \caption{
       Example of the Iman-Conover transformation. Keeping the 2 marginal distributions the same, the rank correlation is increased from the left to right. On the very left, the 2 distributions are independent with no clear pattern. Then, as the correlation positively increases, the relationship becomes stronger.}
   \label{fig:imanconover-example}
\end{figure}

\section{Boosting the active e-value and active e-BH}\label{sec:appendix:boosting}

We can further improve the power of a hypothesis test that uses an active e-value (and consequently active e-BH) by utilizing ideas from stochastic rounding in \citet{xu_more_powerful_2024} to boost the power of the active e-value.

We note that we can define our query indicator to use the active e-value
\begin{align}
T = \ind{U \leq (1 - \gamma F^{-1})_+}.
\end{align}
We can define a modified active e-value as follows:
\begin{align}
    \tilde{E}^{\textnormal{SR}} \coloneqq  (1 - T) \cdot F +
    T \cdot \frac{(1 - \gamma)}{(1 - \gamma F^{-1})_+} \cdot E \label{eq:rounded-active-evalue}
\end{align}
\begin{proposition}
    Let $E$ be an e-value, $F$ be a nonnegative random variable, and $\hat{\alpha} \in [0, 1]$ be a data-dependent threshold that are all mutually arbitrarily dependent. Then
    \begin{align}
        T_{\hat\alpha}(\tilde{E}^\textnormal{SR}) \coloneqq \hat\alpha^{-1}\cdot \indnum\{\tilde{E}^\textnormal{SR} \geq \hat\alpha^{-1}\}
    \end{align} is an e-value. Further, $\tilde{E}^\textnormal{SR} \geq \tilde{E}$ almost surely.
\end{proposition}
\begin{proof}
    We can perform the following decomposition:
    \begin{align}
        \expect\left[
        T_{\hat\alpha}(\tilde{E}^\textnormal{SR})
        \right]
       =
        \expect\left[
            (1 - T) \cdot \hat\alpha^{-1}\indnum\{F \geq \hat\alpha^{-1}\}
        \right]
       +
        \expect\left[
        T \cdot \hat\alpha^{-1} \cdot \ind{\frac{1 - \gamma}{(1 - \gamma F^{-1})_+} E \geq \hat\alpha^{-1}}
        \right]
    \end{align}
    We know that the first summand is bounded by $\gamma$ by definition of $T$ and $\hat\alpha^{-1}\indnum\{F \geq \hat\alpha^{-1}\} \leq F$.

    We make the following derivation when the null hypothesis is true:
    \begin{align}
        &
        \expect\left[T \cdot \hat\alpha^{-1} \cdot \ind{\frac{1 - \gamma}{(1 - \gamma F^{-1})_+} E \geq \hat\alpha^{-1}}\right]
        \\
        &=
        \expect\left[\ind{U \leq (1 - \gamma F^{-1})_+} \cdot \hat\alpha^{-1} \cdot \ind{\frac{1 - \gamma}{(1 - \gamma F^{-1})_+} E \geq \hat\alpha^{-1}}\right]
        \\
        &=
        \expect\left[\hat\alpha^{-1} \cdot \ind{U \leq (1 - \gamma F^{-1})_+}  \cdot \ind{(1 - \gamma F^{-1})_+\leq \hat\alpha (1 - \gamma) \cdot E } \right]
        \\
        &\leq
        \expect\left[\hat\alpha^{-1} \cdot \ind{U \leq \hat\alpha (1 - \gamma) \cdot E } \right]
        \\
        &=
        \expect\left[\hat\alpha^{-1} \cdot \prob{U \leq \hat\alpha (1 - \gamma) \cdot E \mid \hat\alpha, E} \right]
        =
        (1 - \gamma) \expect[E]
        \leq 1 - \gamma.
    \end{align}
    The first inequality is by relaxing the indicator to be for an event that is the superset of the events of the two indicators in the previous line. The last equality follows from $U$ being independent and uniform, and the final inequality follows from $E$ being an e-value. Thus, we see that our e-value of interest is bounded by $\gamma + (1 - \gamma) = 1$,

\end{proof}

\section{Proofs}\label{sec:appendix:proofs}

\subsection{Proof of \Cref{prop: active arb dep pvalue valid}}\label{sec: active arb dep p-value valid proof}

Define $\tilde{P}^* \coloneqq ((1 - \gamma)^{-1} \cdot P) \wedge (U \gamma^{-1})$ where $U \sim \text{Uniform}[0, 1]$ is independent of $P$ and $Q$. Let $T = \ind{U \gamma^{-1} > Q}$, and note that this construction of $T$ satisfies $T \mid Q \sim \text{Bern}(1 - \gamma Q)$. Utilizing this coupling of $\tilde{P}^*$ and $\tilde{P}$ (via $U, T$), we derive the following result when the null hypothesis is true:
    \begin{align}
        \tilde{P} &= \ind{U \gamma^{-1} \leq Q } Q + \ind{U \gamma^{-1}> Q }(1 - \gamma)^{-1} \cdot P\\
                  &\geq \ind{U \gamma^{-1} \leq Q}U\gamma^{-1} + \ind{U \gamma^{-1}> Q }(1 - \gamma)^{-1} \cdot P\\
                  &\geq U\gamma^{-1} \wedge ((1 - \gamma)^{-1} \cdot P) = \tilde{P}^*.
    \end{align}
    The first inequality is by using the fact that $\ind{U \gamma^{-1} \leq Q}U\gamma^{-1} \leq \ind{U \gamma^{-1} \leq Q}Q$. The second inequality is by taking the minimum of the two possible values that the previous line could have taken.
    Thus, $\tilde{P}^* \leq \tilde{P}$. We will now show that $\tilde{P}^*$ is superuniform. For an arbitrary $s \in [0, 1]$, the following holds:
    \begin{align}
        \sprob(\tilde{P}^* \leq s) &= \sprob(U\gamma^{-1} \leq s \text{ or }((1 - \gamma)^{-1} \cdot P) \leq s)\\
                                   &\leq \sprob(U\leq \gamma s) + \sprob(P \leq (1 - \gamma) s) \leq s.
    \end{align}
    The first inequality is by union bound, and the second inequality is because $U \sim \text{Uniform}[0, 1]$ and $P$ is superuniform under the null. Thus, we have shown $\tilde{P}^*$ is a valid p-value, which implies $\tilde{P}$ is a valid p-value as well, concluding our proof.

\subsection{Proof of \Cref{thm:active-bh}}\label{sec:active-bh-proof}

The proof of the PRDN and WNDN case is the result of the following argument. Recall $\tilde{P}^*$ we defined previously --- we can now define
\begin{align}
    \tilde{P}^*_i \coloneqq ((1 - \gamma)^{-1} \cdot P_i) \wedge (U_i \gamma^{-1}) \leq \tilde{P}_i
\end{align} for each $i \in [K]$, where $U_i$ is a uniform random variable on $[0, 1]$ that is independent of all other randomness. Recall that $\tilde{P}_i^*$ is a valid p-value for each $i \in [K]$. 

Let $\tilde{\mathbf{P}}^* \coloneqq (\tilde{P}_1^*, \ldots, \tilde{P}_K^*)$ and  $\mathbf{U} = (U_1, \ldots, U_K)$. We now consider the conditional distribution of $\tilde{\mathbf{P}}^*$ given $\mathbf{U}$, and consider the two cases: when $\mathbf{P}$ is PRDN and when $\mathbf{P}$ is WND.

First we consider the case when $\mathbf{P}$ is PRDN.
Given $\mathbf{U} = \mathbf{u} \coloneqq (u_1, \dots, u_K)$, we have $\tilde{P}_j^* = \phi_j^\mathbf{u}(P_j)$, where $\phi_j^\mathbf{u}(p) \coloneqq ((1 - \gamma)^{-1} \cdot p) \wedge (u_j \gamma^{-1})$ for all $j \in [K]$, making $\tilde{\mathbf{P}}^*$ a comonotone transformation of $\mathbf{P}$.
Thus, $\tilde{\mathbf{P}}^* \mid \mathbf{U} = \mathbf{u}$ is PRDN, for all values of $\mathbf{u} \in [0, 1]^K$ \citep{benjamini_control_false_2001}. This then implies that $\tilde{\mathbf{P}}^*$ is PRDN since the following is true via the tower property of conditional expectation:
\begin{align}
    \prob{\tilde{\mathbf{P}}^* \in D \mid \tilde{P}_i^* \leq s} =\expect[ \prob{\tilde{\mathbf{P}}^* \in D \mid \tilde{P}_i^* \leq s, \mathbf{U}}\mid \tilde{P}_i^* \leq s].
\end{align}

Since $\mathbf{P}$ satisfies PRDN, $D_{\mathbf{u}}$ is increasing, and $(1-\gamma)s \wedge u_i\gamma^{-1}(1-\gamma)$ is non-decreasing in $s$, the right-hand side is non-decreasing in $s$ for each fixed $\mathbf{u}$. Therefore, $\prob{\tilde{\mathbf{P}}^* \in D \mid \tilde{P}_i^* \leq s}$ is non-decreasing in $s$, establishing that $(\tilde{P}_i^*)_{i \in [K]}$ satisfies PRDN.

Now, we consider the case when $\mathbf{P}$ is WNDN. We can make the following derivations for any subset $A \subseteq \Ncal$ of null hypotheses and any $s \in [0, 1]$:
\begin{align}
    \sprob\left(\bigcap_{i \in A}\ \tilde{P}_i^* \leq s\right) 
    &= 
    \sum_{A' \subseteq A} \sprob\left(\bigcap_{i \in A'} P_i \leq (1-\gamma)s, U_i > \gamma s, \bigcap_{i \in A \setminus A'} U_i \leq \gamma s\right) \\
    &= 
    \sum_{A' \subseteq A} \sprob\left(\bigcap_{i \in A'} P_i \leq (1-\gamma)s\right) \cdot \sprob\left( \bigcap_{i \in A'} U_i > \gamma s, \bigcap_{i \in A \setminus A'} U_i \leq \gamma s\right) \\
    &\leq 
    \sum_{A' \subseteq A} \prod_{i \in A'}\sprob(P_i \leq (1-\gamma)s) \cdot \sprob(U_i > \gamma s) \cdot \prod_{i \in A\setminus A'}\sprob(U_i \leq \gamma s)\\
    &= \prod_{i \in A} (\sprob(P_i \leq (1-\gamma)s, U_i > \gamma s) + \sprob(U_i \leq \gamma s))\\
    &= \prod_{i \in A} \sprob(\tilde{P}_i^* \leq s).
\end{align}
Thus, we have shown our desired result that $(\tilde{P}_i^*)_{i \in [K]}$ satisfies WNDN.

Now, we also know that active BH produces a discovery set that is self-consistent w.r.t. $(\tilde{P}_i^*)_{i \in [K]}$, since the active BH procedure produces a discovery set that is self-consistent with $(\tilde{P}_i)_{i \in [K]}$ by definition. 

Thus, we can apply Theorem 3 of \citet{su_fdr_linking_2018} for the PRDN case and Proposition 3.6 of \citet{fischer_online_generalization_2024} for the WNDN case to achieve our desired FDR bound.
Finally, the
FDR control for the arbitrary dependent cases for p-values comes from Theorem 1.3 of \citet{benjamini_control_false_2001}.

\subsection{Proof of \Cref{prop:joint-pvalue-valid}}
\label{sec:joint-pvalue-proof}

Assume the null hypothesis is true. We make the following derivation for the conditional distribution of $\tilde{P}$:
\begin{align}
    \prob{\Tilde{P} \leq s | Q=q}
        &=\prob{ F^{-1}(U|Q=q) \leq s | Q=q} \tag{$\Tilde{P}$ definition}\\
        &= \int_0^1 \prob{ F^{-1}(u|q) \leq s | Q=q} \cdot f_U(u) du \tag{iterated exp.}\\
        &= \int_0^1 \prob{ F^{-1}(u|q) \leq s | Q=q} \cdot 1 du \tag{density of $U$}\\
        &= \int_0^1 \mathbb{I}\{ F^{-1}(u|q) \leq s \} \diff u \tag{constants inside $\mathbb{P}$}\\
        &= \int_0^1 \mathbb{I}\{ u \leq  F(s|q)  \} \diff u \tag{apply $F(\cdot|q)$}\\
        &= F(s|q). \label{eq:joint-proxy-true-cond}
\end{align}
We can now prove the validity of $\tilde{P}$.
\begin{align}
    \prob{\Tilde{P} \leq s} &= \int_0^1 \prob{\Tilde{P} \leq s | Q=q}  f_Q(q) \diff q \tag{iterated exp.}\\
        &= \int_0^1 F(s|q) f_Q(q) \diff q \tag{by \ref{eq:joint-proxy-true-cond}}\\
        &= \int_0^1 \prob{P \leq s | Q=q} f_Q(q) \diff q \tag{$F$ definition}\\
&= \prob{P \leq s} \tag{iterated exp.}\\
        &\leq s. \tag{$P$ is a valid p-value}
\end{align} Thus, we have shown our desired result.
 \end{document}